\documentclass{fundam}

\usepackage{lineno,hyperref}

\usepackage[T1]{fontenc}
\usepackage{graphicx}
\usepackage{flushend}

\usepackage{subfigure}
\usepackage{latexsym}
\usepackage{amssymb}
\usepackage{amsmath}
\usepackage{algorithm}
\usepackage{algorithmic}
\usepackage{url}
\usepackage{rotating}
\usepackage{color}
\definecolor{blank}{rgb}{0.7,0.7,0.7}
\definecolor{gray}{rgb}{0.7,0.7,0.7}
\usepackage{captcont}
\usepackage{mathptmx}
\usepackage{listings}
\usepackage{enumerate}

\makeatletter
\newenvironment{prog}{\vspace{1.0ex}\par
	\obeylines\@vobeyspaces\tt}{\vspace{1.0ex}\noindent } \makeatother
\newcommand{\startprog}{\begin{prog}}
	\newcommand{\stopprog}{\end{prog}\noindent}

\long\def\comment#1{}

\def \tuple#1{\langle #1 \rangle}

\def\defemb#1#2{\expandafter\def\csname #1\endcsname
	{\relax\ifmmode #2\else\hbox{$#2$}\fi}}

\DeclareSymbolFont{symbolsC}{U}{pxsyc}{m}{n}
\SetSymbolFont{symbolsC}{bold}{U}{pxsyc}{bx}{n}
\def\re@DeclareMathSymbol#1#2#3#4{%
	\let#1=\undefined
	\DeclareMathSymbol{#1}{#2}{#3}{#4}}
\re@DeclareMathSymbol{\dashleftrightarrow}{\mathrel}{symbolsC}{101}

\defemb{cA}{{\mathcal A}}
\defemb{cB}{{\mathcal B}}
\defemb{cC}{{\mathcal C}}
\defemb{cD}{{\mathcal D}}
\defemb{cE}{{\mathcal E}}
\defemb{cF}{{\mathcal F}}
\defemb{cG}{{\mathcal G}}
\defemb{cH}{{\mathcal H}}
\defemb{cI}{{\mathcal I}}
\defemb{cJ}{{\mathcal J}}
\defemb{cL}{{\mathcal L}}
\defemb{cM}{{\mathcal M}}
\defemb{cN}{{\mathcal N}}
\defemb{cO}{{\mathcal O}}
\defemb{cP}{{\mathcal P}}
\defemb{cR}{{\mathcal R}}
\defemb{cS}{{\mathcal S}}
\defemb{cT}{{\mathcal T}}
\defemb{cU}{{\mathcal U}}
\defemb{cV}{{\mathcal V}}
\defemb{cX}{{\mathcal X}}
\defemb{cZ}{{\mathcal Z}}

\newcommand{\bslice}{\mathsf{b\_slice}}
\newcommand{\fslice}{\mathsf{f\_slice}}
\newcommand{\filter}{\mathsf{filter}}
\newcommand{\fire}[1]{\stackrel{#1}{\longrightarrow}}

\newcommand{\Nat}{{\mathbb N}}

\newif\ifpaperVersion
\paperVersionfalse
\ifpaperVersion
\newcommand{\ignore}[1]{}
\newcommand{\deleted}[1]{}
\newcommand{\added}[1]{}
\newcommand{\pending}[1]{}
\newcommand{\done}[1]{}
\newcommand{\doubt}[1]{}
\newcommand{\josep}[1]{}
\newcommand{\tama}[1]{}
\newcommand{\marisa}[1]{}
\newcommand{\javi}[1]{}
\else
\definecolor{ignoreColor}{rgb}{1,0.5,0}
\definecolor{pendingColor}{rgb}{0.9,0.0,0.2}
\definecolor{doneColor}{rgb}{0.7,0.2,0.7}
\definecolor{doubtColor}{rgb}{0.6,0.6,0.4}
\definecolor{marisaColor}{rgb}{0.2,0.6,0.6}
\definecolor{javiColor}{rgb}{0.6,0.2,0.6}
\definecolor{josepColor}{rgb}{1.0,0.5,0.0}
\definecolor{tamaColor}{rgb}{0.3,0.5,1.0}

\newcommand{\ignore}[1]{\textcolor{ignoreColor}{\#Ignored: #1}}
\newcommand{\deleted}[1]{\textcolor{red}{\#Deleted: #1}}
\newcommand{\added}[1]{\textcolor{blue}{\#Added: #1}}
\newcommand{\pending}[1]{\textcolor{pendingColor}{\textbf{\#Pending: #1}}}
\newcommand{\done}[1]{\textcolor{doneColor}{\#Done: #1}}
\newcommand{\doubt}[1]{\textcolor{doubtColor}{\#Doubt: #1}}

\newcommand{\marisa}[1]{\todo[color=marisaColor!20,bordercolor=marisaColor,linecolor=marisaColor,size=\scriptsize]{\#MMM: #1}}
\newcommand{\javi}[1]{\todo[color=javiColor!20,bordercolor=javiColor,linecolor=javiColor,size=\scriptsize]{\#JAV: #1}}
\newcommand{\josep}[1]{\todo[color=josepColor!20,bordercolor=josepColor,linecolor=josepColor,size=\scriptsize]{\#JJJ: #1}}
	
\newcommand{\tama}[1]{\todo[color=tamaColor!20,bordercolor=tamaColor,linecolor=tamaColor,size=\scriptsize]{\#TTT: #1}}
\fi

\begin{document}
	

\setcounter{page}{239}
\publyear{22}
\papernumber{2148}
\volume{188}
\issue{4}

\finalVersionForARXIV

		\title{Maximal and Minimal Dynamic Petri Net Slicing}

	\author{M. Llorens\thanks{This work has been partially supported by the EU (FEDER) and the Spanish
			MCI/AEI under grant PID2019-104735RB-C41 and by the European Union's Horizon 2020 research and
			innovation programme under grant agreement No 952215 (Tailor).},
           ~J. Oliver\thanksas{1},  ~J. Silva\thanksas{1}\thanks{Address for  correspondence:  VRAIN, Departamento
            de Sistemas Inform\'aticos y Computaci\'on. Universitat Polit\`ecnica de Val\`encia. Valencia, Spain. \newline \newline
                    \vspace*{-6mm}{\scriptsize{Received November 2022; \ accepted  April  2023.}}}, \ and  ~S. Tamarit\thanksas{1}
         \\
  VRAIN, Departamento de Sistemas Inform\'aticos y Computaci\'on\\
  Universitat Polit\`ecnica de Val\`encia\\
  Valencia, Spain\\
  \{mllorens,~fjoliver,~jsilva\}@dsic.upv.es
    }

	\maketitle
	
	\runninghead{M. Llorens et al.}{Maximal and Minimal Dynamic Petri Net Slicing}
	
     \vspace*{-2mm}
	\begin{abstract}~

		{\bf Context:}
		Petri net slicing is a technique to reduce the size of a Petri net to ease the analysis or understanding of the original Petri net.

		{\bf Objective:}
		Presenting two new Petri net slicing algorithms to isolate those places and transitions of a Petri net (the slice) that may contribute tokens to one or more places given (the slicing criterion).

		{\bf Method:} The two algorithms proposed are formalized. The maximality of the first algorithm and the minimality of the second algorithm are formally proven. Both algorithms together with three other state-of-the-art algorithms have been implemented and integrated into a single tool so that we have been able to carry out a fair empirical evaluation.

		{\bf Results:}
		Besides the two new Petri net slicing algorithms, a public, free, and open-source implementation of five algorithms is reported. The results of an empirical evaluation of the new algorithms and the slices they produce are also presented.

		{\bf Conclusions:}
		The first algorithm collects all places and transitions that may contribute tokens (in \emph{any} computation) to the slicing criterion, while the second algorithm collects the places and transitions needed to fire the shortest transition sequence that contributes tokens to \emph{some} place in the slicing criterion. Therefore, the net computed by the first algorithm can reproduce any computation that contributes tokens to any place of interest. In contrast, the second algorithm loses this possibility, but it often produces a much more reduced subnet (which still can reproduce some computations that contribute tokens to some places of interest). The first algorithm is proven maximal, and the second one is proven minimal.
		
	\end{abstract}

	\section{Introduction}\label{intro}
	
	\emph{Program slicing} \cite{Tip95,Sil12} is a technique to extract from a given program all the statements that are influenced by (forward slicing) or that do influence (backward slicing) a specified point of interest. This point is referred to as \emph{slicing criterion}. A slicing criterion is often a pair composed of a program line and a variable of interest. The slice associated with this slicing criterion is the set of statements that influence or are influenced by the specified variable.
	
	\begin{example}
		Consider the program on the left. The code on the right is the backward slice of this program w.r.t. the slicing criterion {\tt <4,n>}.
		
		\smallskip
		{\tt
			\small
			\hspace{0.5cm} (1) read(n);                      \hspace{1.95cm} (1) read(n);\\
			\indent \hspace{0.5cm} (2) read(m);\\
			\indent \hspace{0.5cm} (3) if (n>10)          \hspace{1.8cm} (3) if (n>10)\\
			\indent \hspace{0.5cm} (4) then write(n);   \hspace{0.9cm} (4) then write(n);\\
			\indent \hspace{0.5cm} (5) else write(n+m);
		}
		\smallskip
		
		\noindent This slice contains the statements that can influence the values computed for variable {\tt n} at line 4.
	\end{example}
	
	Program slicing 	was adapted to Petri nets 
	for the first time in \cite{ChW87}.
	\emph{Petri net slicing} \cite{Khan13Survey,Khan18Survey} allows us to extract from a Petri net all places and transitions that are related to a specified \emph{slicing criterion}. One common approach is to define the notion of slicing criterion as a set of places in a given marked Petri net. Then, the slice is computed by extracting all those places and transitions that are associated with the slicing criterion in some way. Different slicing algorithms do different jobs, i.e. preserve different properties and hence reduce a net in different ways. For instance, a slice could be defined as the subnet that can contribute tokens to the slicing criterion. All slicing algorithms can be classified as static or dynamic. Dynamic slicing algorithms take into account the initial marking of the net to produce the slice. In contrast, static slicing algorithms ignore the initial marking. 
	
	\begin{example}
		\label{ex_motiv}
		Consider the following Petri net and the slicing criterion $\{p_4\}$ (the place coloured in grey).
		
		\begin{center}
			\includegraphics[scale=0.45]{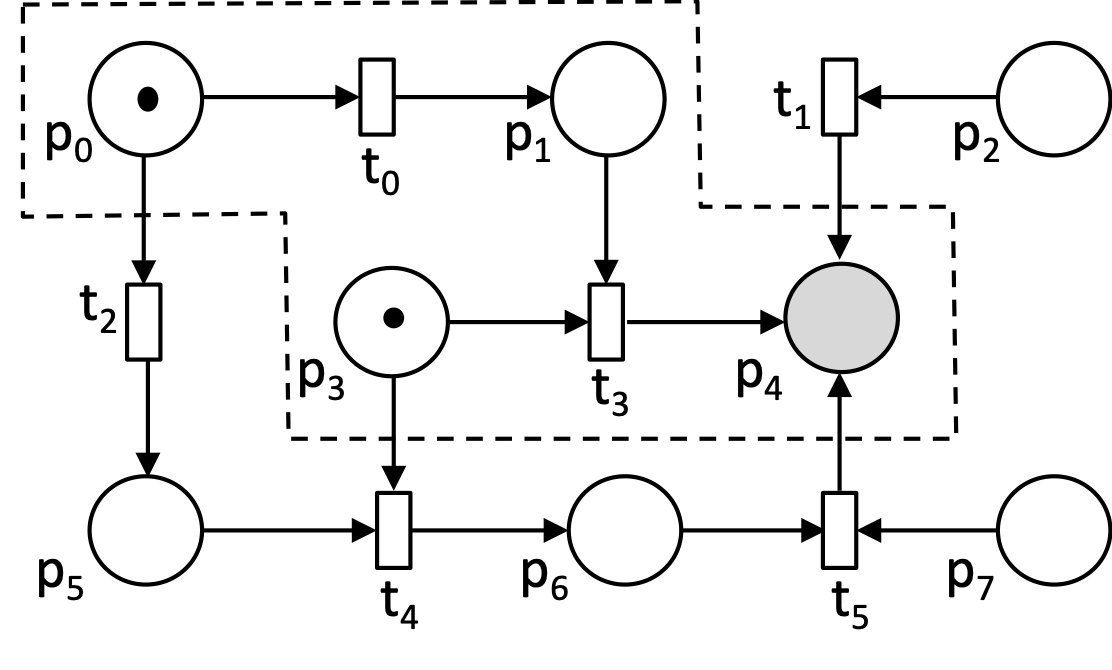}
		\end{center}
		
		\noindent The dynamic backward slice computed is the subnet inside the dashed area. This slice contains the only places and transitions that can contribute tokens to the slicing criterion from the initial marking. Note that ${\tt t_1}$, ${\tt t_2}$, ${\tt t_4}$, and ${\tt t_5}$ cannot contribute tokens in the current state (i.e. with this marking).
	\end{example}
	
	One of the main uses of Petri net slicing is in model checking, as a preprocessing stage to produce a reduced Petri net \cite{Khan13Survey,Khan18Survey,Khan13,LCKK00,Rak12,Yu15}. Reducing the state space explosion problem increases the scalability of the analyses. However, not all slicing algorithms reduce the space explosion problem when they reduce the Petri net, they can even increase it. This was experimentally shown with various static slicers in \cite{Dav20} (see the discussion in Sect. V), so slicing must be used with a clear objective and a careful design.
	
	Among other uses of Petri net slicing we have debugging \cite{Llo08,Yu15}, comprehension of Petri nets \cite{Llo08,Dav20b}, modularization \cite{Dav20b}, and enhancing of reachability analyses \cite{Khan13,LCKK00,Rak12}. Another approach closely related is
	activity-oriented Petri Nets, a methodology to minimize the size of a Petri net. It has been used to reduce the size of Petri net models involving many resources \cite{WanY15,IndT18,YanLAA18} and tested for its applicability for model checking in \cite{Dav20} together with other slicing algorithms (\cite{Rak08,Rak12,DavR18}).
	
	Over the years, several different definitions of Petri net slicing have emerged (including static and dynamic slicing of Petri nets) and alternative approaches to their computation \cite{Rak12,Yu15,Llo08}. 
	For instance, in Example~\ref{ex_motiv}, the transitions ${\tt t_2}$ and ${\tt t_4}$ cannot contribute tokens to the slicing criterion, but they could prevent the tokens in ${\tt p_0}$ and ${\tt p_3}$, respectively, to reach the slicing criterion. Therefore, they are included in some notions of Petri net slicing. The most extended algorithms were reviewed in \cite{Khan13Survey,Khan18Survey}, showing that sometimes they are complementary, and for some applications, different algorithms can be used.
	In this work, we present two Petri net slicing algorithms that complement the state of the art.

	\subsection{Motivation}
	
	The algorithms proposed in this work are useful for Petri net simplification to enhance verification and analysis.
	In particular, we provide a notion of minimal slice which is not achieved by any current algorithm.
	However, the main motivation of this work is not model checking.
	In contrast, we want to provide a new notion of Petri net slicing that is especially useful for debugging and specialization. Thus, they can be used during Petri net construction.
	
	Our first algorithm improves the behaviour and the efficiency of the algorithm by Llorens et al. \cite{Llo08}. It can be used in debugging: when we reach a particular state and we detect a place with a token that should not be there, or just with more tokens than it should have, then we can produce a slice that only contains the part of the Petri net that contributed tokens to that place. Therefore, the bug must be inside the slice.
	For instance, in Example~\ref{ex_motiv}, the part of the net that is responsible for a (possibly wrong) token in $p_4$ is the slice.
	
	Our second algorithm produces minimal slices, as defined in Definition \ref{minimalslice-def}, and it often produces smaller slices than all the other algorithms. It can be particularly useful for program specialization: when we want to extract a component that fires a specific transition from a given state, then we can produce a slice that contains the subnet that fires the desired transition from the given state with the minimum set of transition firings.
	For instance, in Example~\ref{ex_motiv}, the slice is a subcomponent that can be reused in another net, or, e.g., used to understand one specific part of the net.
	
	These two slices are different from those computed by other current algorithms. This is illustrated in Example~\ref{example-motiv}.
	
	\begin{example}\label{example-motiv}
		Figure~\ref{fig:initialPNandRakowCTL} shows a Petri net where the set of places coloured in grey ($\{p_6,p_9\}$) is the slicing criterion.
		We computed a slice of this Petri net with five different algorithms producing the nets in Figures \ref{fig:initialPNandRakowCTL}, \ref{fig:rakowSafety}, \ref{fig:llorensetal}, \ref{fig:llorensetalimproved} and \ref{fig:yuetal}.
		The Petri net in this example and its slices is an interesting contribution because all slices computed are pairwise different (the slices often coincide). 
	\end{example}
	
	\begin{figure}[h!]
		\centering
		\subfigure[\label{fig:initialPNandRakowCTL}Example PN and CTL$^{*}_{-x} slice$ \cite{Rak12}]{\includegraphics[scale=0.45]{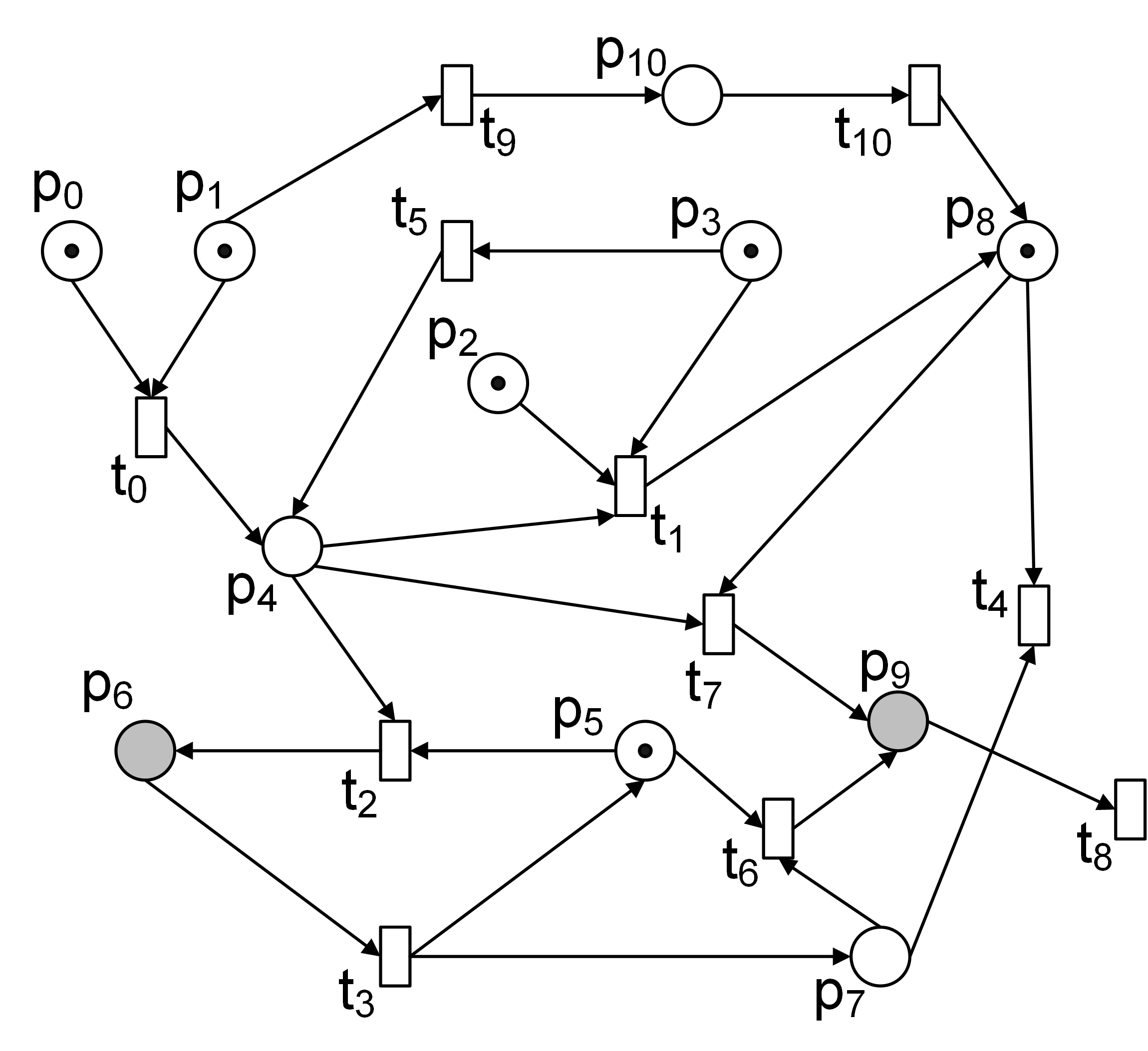}~~~}
		\quad
		\subfigure[\label{fig:rakowSafety}Safety slice \cite{Rak12}]{\includegraphics[scale=0.45]{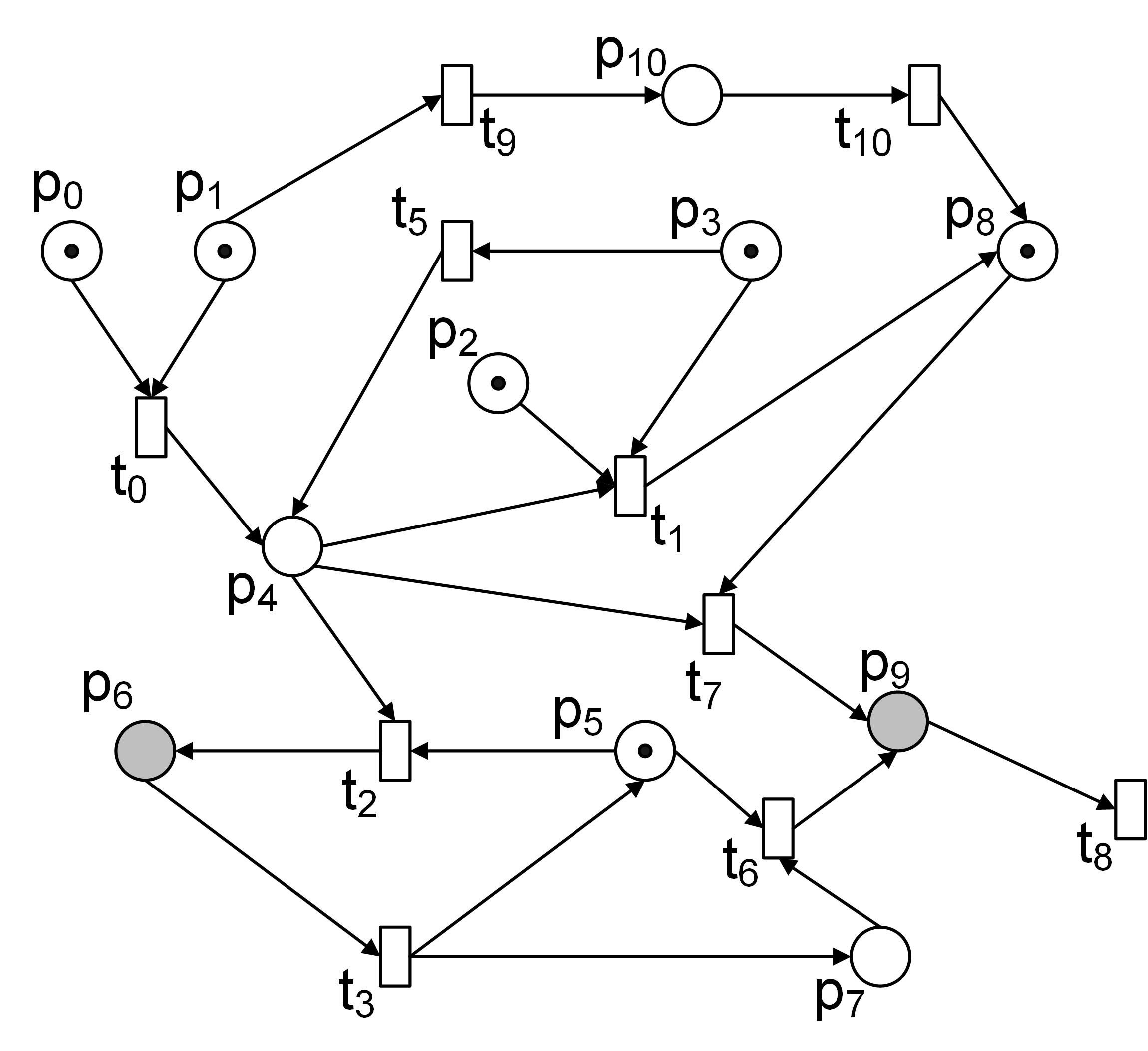}}
		\\
		\subfigure[\label{fig:llorensetal}Slice by our Algorithm 1]{\includegraphics[scale=0.45]{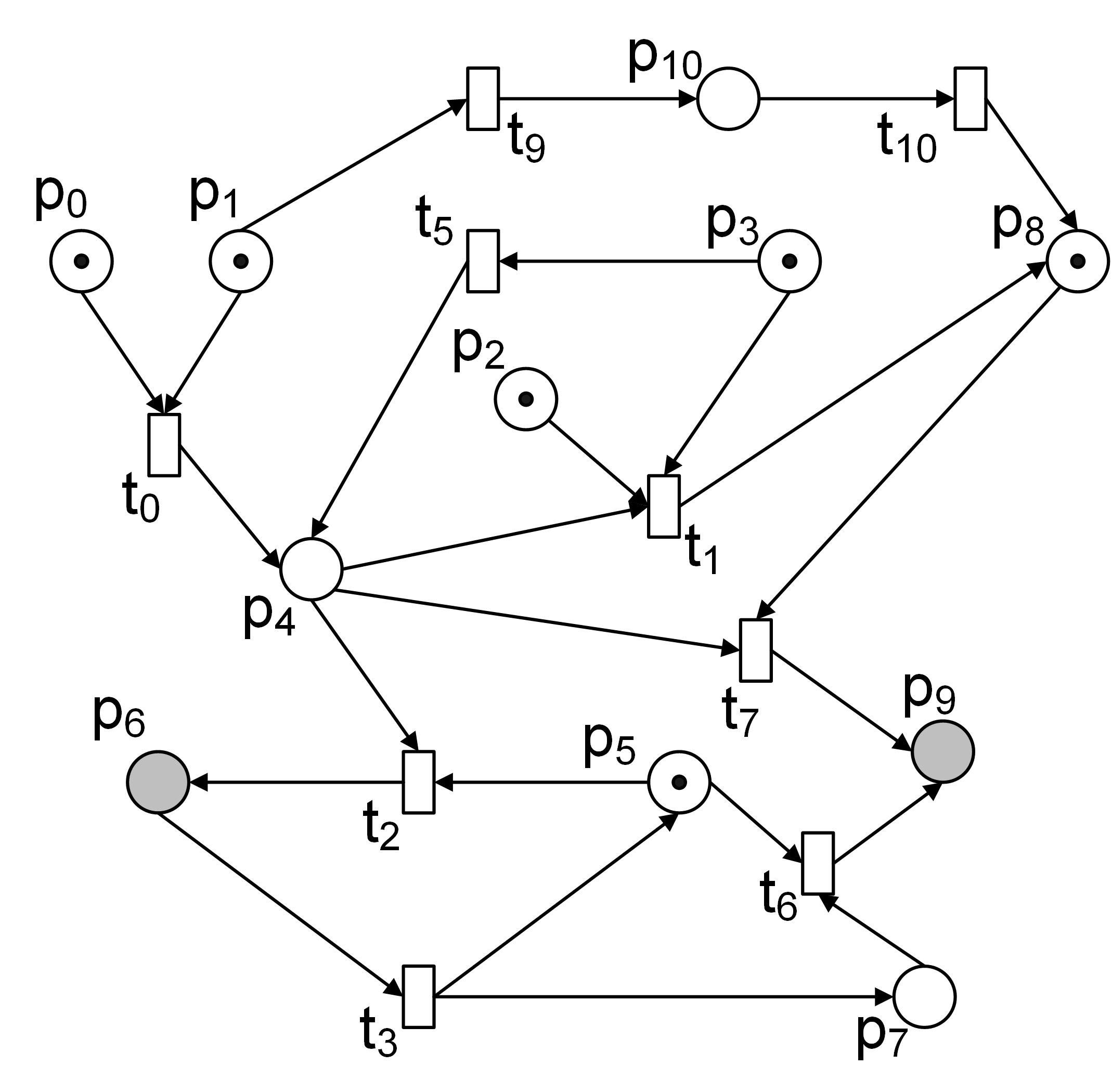}}\quad 
		\subfigure[\label{fig:llorensetalimproved}Slice by our Algorithm 2]{\includegraphics[scale=0.45]{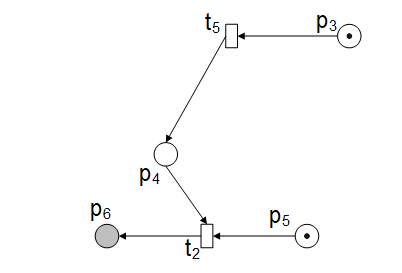}}\quad
		\subfigure[\label{fig:yuetal}Slice by Yu et al. \cite{Yu15}]{\includegraphics[scale=0.45]{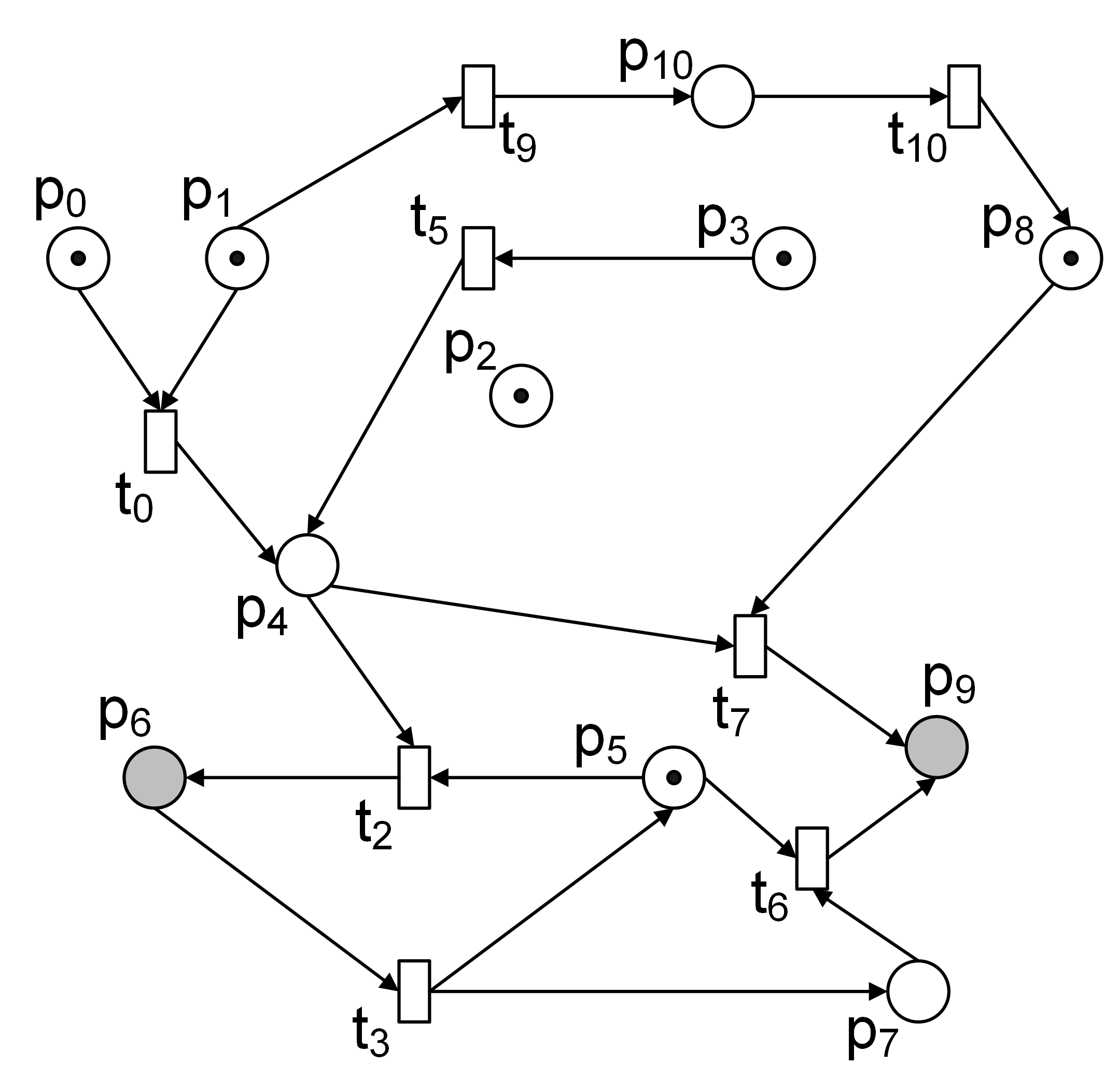}}
		\caption{Slices (a), (b), (c), (d), and (e) extracted from a Petri net (a).}\label{fig:slicingResults}
	\end{figure}
	
	The main motivation of this work is, on the one hand, to formalize the two new Petri net slicing algorithms and prove their maximality and minimality; and, on the other hand, to empirically evaluate them so that we can measure the size of the computed slices and their performance compared to other state-of-the-art algorithms. 
	
	\subsection{Contributions}
	
	The main contributions of this work are the following:
	\begin{enumerate}[(i)]
		\item Two Petri net slicing algorithms.
		They are properly formalized so that we prove some of their properties (maximality and minimality).
		\item A public, free, and open-source implementation of the algorithms.
		\item An empirical evaluation of the new algorithms and the slices that they produce.
	\end{enumerate}

	\subsection{Structure of the paper}

	First, Section~\ref{sec_Prelims} introduces some preliminary definitions needed to establish a theoretical basis to formalize the new algorithms.
	Next, Section~\ref{sec_Algorithms} presents the algorithms and formally proves some of their properties.
	Section~\ref{sec_stateArt} presents the related work, describing other related approaches for Petri net slicing.
	Section~\ref{PN-Slicer} presents the implementation, a tool called \texttt{pn$\_$slicer}.
	In Section~\ref{Sec-Evaluation} we empirically evaluate the algorithms implemented.
	Finally, Section~\ref{Sec-Conclusions} concludes.
	
	\section{Petri net slicing}\label{sec_Prelims}
	\subsection{Petri nets and subnets}
	
	We start with the definition of a Petri net \cite{Murata89,Peterson81,Reisig13}. A \emph{Petri net} may be viewed as a directed bipartite graph with an initial state usually known as \emph{initial marking}. The graph is composed of \emph{places} (represented by circles) and \emph{transitions} (represented by rectangles), which are connected by means of directed \emph{arcs} labelled with a positive integer that represents their \emph{weight}. The arcs can only connect transitions to places and vice-versa. Places are assigned a non-negative integer known as \emph{marking}. This marking is the number of \emph{tokens} contained in the places and is graphically represented with small black circles in the places. The overall marking of the graph (for all places) is a \emph{state} of the system. Formally,
	
	\begin{definition} \label{pn}
		A \emph{Petri net} \cite{Murata89,Peterson81,Reisig13} is a tuple
		$\cN=(P,T,F)$, where:
		\begin{itemize}
			\item $P$ is a set of \emph{places}.
			\item $T$ is a set of \emph{transitions}, such that $P \cap
			T=\emptyset$ $\wedge$ $P \cup T\neq\emptyset$.
			\item $F$ is the \emph{flow relation} 
			that assigns weights to arcs:\\ $F \subseteq (P \times T) \ \cup \ (T \times P)
			\rightarrow \Nat\backslash\{0\}$.
		\end{itemize}
	\end{definition}
	
	A Petri net is said to be \emph{ordinary} if all of its arc weights are 1's.
	
	\begin{definition} \label{markedPN}
		
		A \emph{marking} $M: P \rightarrow \Nat$ of a Petri net is defined over the set of
		places $P$. For each place $p \in P$ we let $M(p)$ denote the number
		of tokens contained in $p$.
		
		A \emph{marked Petri net} $\Sigma$ is a pair $(\cN,M)$ where $\cN$ is
		a Petri net and $M$ is a marking. We represent the \emph{initial
			marking} of the net by $M_0$.
	\end{definition}

	In the following, given a marking $M$ and a set of places $Q$, we
	denote by $M|_Q$ the \emph{restriction} of $M$ over $Q$, i.e., $M|_Q(p) =
	M(p)$ for all $p\in Q$ and $M|_Q$ is undefined otherwise.

	Given a Petri net $\cN=(P,T,F)$, we say that a place $p\in P$ is an
	\emph{input (resp.\ output) place} of a transition $t\in T$ iff
	there is an \emph{input (resp.\ output) arc} from $p$ to $t$ (resp.\
	from $t$ to $p$). Given a transition $t\in T$, we denote by
	$^\bullet t$ and $t^\bullet$ the sets of all input and output places
	of $t$, respectively. Analogously, given a place $p\in P$, we denote
	$^\bullet p$ and $p^\bullet$ the sets of all input and output
	transitions of $p$, respectively.
	
	\begin{definition}\label{pn_t}
		Let $\Sigma = (\cN,M)$ be a marked Petri net, with $\cN=(P,T,F)$.
		We say that a transition $t\in T$ is \emph{enabled} in $M$, in symbols $M\fire{t}$, iff for
		each input place $p\in P$ of $t$, we have $M(p) \geq F(p,t)$.
		
		A transition may only be fired if it is enabled. The \emph{firing} of an enabled transition $t$ in a marking $M$
		eliminates $F(p,t)$ tokens from each input place $p \in {}^\bullet
		t$ and adds $F(t,p')$ tokens to each output place $p' \in
		t^\bullet$, producing a new marking $M'$, in symbols $M \fire{t}
		M'$.
	\end{definition}

	We say that a marking $M_{n}$ is \emph{reachable} from an initial
	marking $M_{0}$ if there exists a \emph{firing sequence} $\sigma = t_1
	t_2 \ldots t_n$ such that $M_{0} \fire{t_1} M_1 \fire{t_2} \ldots
	\fire{t_n} M_{n}$. 
	In this case, we say that $M_n$ is reachable from $M_0$ through
	$\sigma$, in symbols $M_{0} \fire{\sigma} M_{n}$. This notion includes
	the empty sequence $\epsilon$; we have $M \fire{\epsilon} M$ for any
	marking $M$.  We say that a firing sequence is \emph{initial} if it
	is enabled at the initial marking.	
	
	We say that $\sigma'$ is a \emph{subsequence} of a firing
	sequence $\sigma$ w.r.t.\ a set of transitions $T$ if $\sigma'$
	contains all transition firings in $\sigma$ that correspond to transitions in $T$ and in the
	same order.

    \begin{definition}\label{covers}
     Let $M$ and $M'$ be markings of a Petri net $\cN=(P,T,F)$. $M'$ \emph{covers} $M$ if $M(p) \leq M'(p) \forall p \in P$, in symbols $M \leq M'$.
     If furthermore $M \neq M'$, we say that $M < M'$, and $M'$ \emph{strictly covers} $M$. If neither marking covers the other, they are \emph{incomparable}.
    \end{definition}

    \begin{lemma}\label{monotonicity-lemma}
    Petri nets are \emph{strictly monotonic}. Let $M$ and $M_1$ be markings of a Petri net such that $M < M_1$, and a firing sequence $\sigma$ such that $M \fire{\sigma} M'$. Then, $M_1 \fire{\sigma} M_1'$ and $M' < M_1'$.
    \end{lemma}

	The set of reachable markings or \emph{reachability set} is the set of all possible markings that are reachable from an initial
	marking $M_{0}$ in a marked Petri net $\Sigma = (\cN,M_{0})$, denoted by $R(\cN,M_{0})$ (or simply by $R(M_{0})$ when $\cN$ is clear
	from the context).

    To list all the markings in $R(\cN,M_{0})$ we can construct a \emph{reachability tree}. The construction of the reachability tree consists of taking $M_{0}$ as the root of the tree and firing all the enabled transitions in $M_{0}$. This leads to new markings that enable other transitions. Taking each of those new markings as a new root, all reachable markings can be recursively generated.
    \begin{definition}\label{def:reachabilitytree}
    A \emph{reachability tree} $\cT = (V, E)$ of a marked Petri net $\Sigma = (\cN,M_0)$, with $\cN=(P,T,F)$, is an edge-labelled directed rooted tree, where nodes $V$ are markings $\in R(M_0)$, the root is $M_0$ and edges $E \subseteq (V \times T \times V)$ such that $(M, t, M') \in E$ if $M\fire{t}M'$, with $M, M' \in V$ and $t \in T$.

    A \emph{path} between two nodes $v, v' \in V$ of $\cT$, in symbols $RTpath(v, v')$, is a sequence of distinct edges, interleaved with nodes, that \emph{lead from} $v$ \emph{to} $v'$.
    Given a node $v \in V$, the nodes with a path to $v$ are the \emph{ancestors} of $v$, and the nodes with a path from $v$ are the \emph{successors} of $v$.
    \end{definition}

	We use the following notion of \emph{subnet} to define Petri net slicing
	(roughly speaking, we identify a slice with a subnet).
	First, given $(P' \times T') \ \cup \ (T' \times P') \subseteq (P
	\times T) \ \cup \ (T \times P)$, we say that a flow relation $F':(P'
	\times T') \ \cup \ (T' \times P') \rightarrow \Nat$ is a \emph{restriction} of
	another flow relation $F:(P \times T) \ \cup \ (T \times P) \rightarrow
	\Nat$ over $P'$ and $T'$, in symbols $F|_{(P',T')}$, if $F'$ is
	defined as follows: $F'(x,y) = F(x,y)$ if $(x,y)\in (P' \times T') \
	\cup \ (T' \times P')$. 
	
	\begin{definition}\label{subnet}\cite{DesEsp95}
		A \emph{subnet} $\cN'=(P',T',F')$ of a Petri net $\cN=(P,T,F)$, denoted $\cN' \subseteq \cN$, is a
		Petri net such that $P'\subseteq P$, $T'\subseteq T$ and $F'$ is a
		restriction of $F$ over $P'$ and $ T'$, i.e., $F' = F|_{(P',T')}$.
	\end{definition}

	\subsection{Slicing Petri nets}
	
	In this section, we formalize our notion of Petri net slicing, giving a formal definition for slicing criterion and slice in the context of Petri nets. Roughly, a slicing criterion is composed of a set of places and an initial marking. With a slicing criterion we can compute slices of a Petri net, which are a subnet that preserves at least one firing sequence from the initial marking that contributes tokens to the places of the slicing criterion. In the following, we formalize these ideas and provide a definition for minimal and maximal slices.

	Besides static or dynamic, slicing algorithms are classified as forwards/backwards:
	a forward slice is formed from those places and transitions that can be influenced by the slicing criterion.
	The opposite is a backward slice: those places and transitions that can influence the slicing criterion.
	Our algorithms produce backward slices.
	
	Before formally defining the new algorithms for Petri net slicing, we first provide a formal definition of \emph{slicing criterion} and \emph{slice} of a Petri net.
	
	\begin{definition} \label{slicing-criterion}
		Let $\cN=(P,T,F)$ be a Petri net. A \emph{slicing criterion} for $\cN$ is
		a pair $\tuple{M_0,Q}$ where $M_0$ is an initial marking for $\cN$
		and $Q\subseteq P$ is a set of places.
	\end{definition}
	
	For the definition of a slice, we need to introduce the notion of \emph{increasing firing sequence}, which, roughly, is the firing of $\sigma$, a sequence of transitions that eventually increases the number of tokens in some place $p$ of the slicing criterion (i.e. $F(t_n,p) > F(p,t_n)$ where $t_n$ is the last transition of $\sigma$).

	\begin{definition} \label{increasing-sequence-def} Let $\cN = (P,T,F)$ be a Petri
		net and let $\tuple{M_0,Q}$ be a slicing criterion for $\cN$.
		An \emph{increasing firing sequence} of $\cN$ w.r.t. $\tuple{M_0,Q}$ is a firing sequence $\sigma = t_1 \ldots t_n$
		with $M_0 \fire{t_1} \ldots \fire{t_{n-1}} M_{n-1}
		\fire{t_n} M_n$ such that $M_{n-1}(p) < M_n(p)$ for some $p\in Q$.
	\end{definition}
	
	Based on the definition of increasing firing sequence, we can provide a notion of slice.
	
	\begin{definition} \label{slice1-def} Let $\cN = (P, T, F)$ be a Petri
		net and let $\tuple{M_0,Q}$ be a slicing criterion for $\cN$. Given
		a Petri net $\cN' = (P',T',F')$, we say that $\cN'$ is a \emph{slice} of
		$\cN$ w.r.t.\ $\tuple{M_0,Q}$ if the following conditions hold:
		\begin{itemize}
			\item the Petri net $\cN'$ is a subnet of $\cN$, such that $\not \exists p \in$ $^\bullet t, t \in T~|~t \in T' \wedge p \notin P'$.
			\item there exists an increasing firing sequence $\sigma'$ in $\cN'$ w.r.t. $\tuple{M_0,Q}$ such that $\sigma'$ is a subsequence of $\sigma$, where $\sigma$ is an increasing firing sequence in $\cN$ w.r.t. $\tuple{M_0,Q}$.
		\end{itemize}
	\end{definition}
	
	This definition of slice forces all transitions in the slice to keep their input places. This avoids that a non-source transition in the original net becomes a source transition in the slice producing a firing sequence that was not possible in the original net. Moreover, this notion of slice is very flexible because it only requires the existence of one increasing firing sequence in the slice.
	We can make the definition more restrictive if we require that all increasing firing sequences of the original net have a counterpart in the slice. We call this kind of slice \emph{maximal} slice.
	
	\begin{definition} \label{maximalslice-def} Let $\cN$ be a Petri
		net and let $\tuple{M_0,Q}$ be a slicing criterion for $\cN$. Given
		a Petri net $\cN'$, we say that $\cN'$ is a \emph{maximal} slice of
		$\cN$ w.r.t.\ $\tuple{M_0,Q}$ if the following conditions hold:
		\begin{itemize}
			\item the Petri net $\cN'$ is a slice of $\cN$ w.r.t.\ $\tuple{M_0,Q}$,
			\item for each increasing firing sequence $\sigma$ in $\cN$ w.r.t. $\tuple{M_0,Q}$, there exists an increasing firing sequence $\sigma'$ in $\cN'$ w.r.t. $\tuple{M_0,Q}$ such that $\sigma'$ is a subsequence of $\sigma$.
		\end{itemize}
	\end{definition}

	Finally, we can also provide a definition of \emph{minimal} slice, which is a slice that can fire the shortest increasing firing sequence.

	\begin{definition} \label{minimalslice-def} Let $\cN$ be a Petri
		net and let $\tuple{M_0,Q}$ be a slicing criterion for $\cN$. Given
		a Petri net $\cN'=(P',T',F')$, we say that $\cN'$ is a \emph{minimal} slice of
		$\cN$ if the following conditions hold:
		\begin{itemize}
			\item the Petri net $\cN'$ is a slice of $\cN$ w.r.t.\ $\tuple{M_0,Q}$,

            \item the Petri net $\cN'$ only contains the places and transitions needed to fire an increasing firing sequence $\sigma' = t'_1 \ldots t'_n$ in $\cN'$ w.r.t. $\tuple{M_0,Q}$; and there does not exist a slice $\cN''=(P'',T'',F'')$ of $\cN$ w.r.t.\ $\tuple{M_0,Q}$ with an increasing firing sequence $\sigma'' = t''_1 \ldots t''_m$ such that $m<n$.
		\end{itemize}
	\end{definition}
	
	It is important to clarify that a maximal slice is not necessarily the biggest slice (considering the size of a Peri net as its number of transitions). The biggest slice (if it exists) is always the original Petri net because, according to Definition~\ref{slice1-def}. A maximal slice, however, is often smaller than the original Petri net.
On the other hand, the minimal slice of a Petri net is not necessarily the smallest one and it is not necessarily unique. There can coexist many different minimal slices of a given Petri net, all of them with a (same size) shortest increasing firing sequence.

\begin{example}\label{example-minimal-smallest}
The following Petri net clearly shows that the minimal slice is not necessarily the smallest one:

\vspace*{-2mm}
\begin{center}
	\includegraphics[scale=0.48]{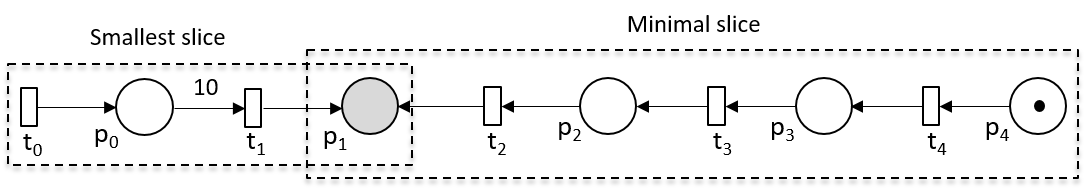}
\end{center}

The definition and algorithm needed to compute the smallest slice is given in Section \ref{sec:minimal-slice2}.
\end{example}

It is worth to note that the above definitions are a refinement and a generalization over previous definitions of Petri net dynamic slice. First, the standard notion of dynamic slice (the one corresponding to the standard definition of dynamic slice in program slicing \cite{KL88}), defined in \cite{Llo08}, forces the slice to include \emph{all} the paths that, from the initial marking, can contribute tokens to \emph{some} place of the slicing criterion. This corresponds to our definition of maximal slice (\ref{maximalslice-def}). Another definition is the one given by Yu et al. \cite{Yu15}, which only includes in the slice \emph{one} single path (which is not necessarily the minimal one). Our definition of minimal slice (\ref{minimalslice-def}) corresponds to the definition of Yu et al. but it forces the increasing firing sequence of the slice to be minimal. Therefore, our definition of dynamic slice (\ref{slice1-def}) generalizes all previous definitions because the slices in \cite{Llo08} and \cite{Yu15} (and also maximal and minimal slices) are particular instances of that definition.

	\section{Two algorithms for Petri net slicing}\label{sec_Algorithms}
	
	In this section, we propose two algorithms to compute, respectively, maximal and minimal slices.
	The Petri net slices computed by the two algorithms have different purposes and exhibit different properties,
	but both of them are \emph{dynamic} because they consider an initial marking in the Petri net.
	
	\begin{description}
		\item Algorithm~\ref{alg:Slicing-PN}: It computes the union of all paths that can contribute tokens to any place in the slicing criterion. Therefore, all firing sequences in the original net that can contribute tokens to the slicing criterion are preserved in the slice.
		\item Algorithm~\ref{alg:Slicing-PN2}: It computes the path that can contribute tokens to some place in the slicing criterion by firing the minimum number of transitions. Therefore, at least one firing sequence in the original net that can contribute tokens to some place in the slicing criterion is preserved in the slice.
	\end{description}
	
	The difference between both algorithms can be seen in Figures~\ref{fig:llorensetal} and \ref{fig:llorensetalimproved}. Figure~\ref{fig:llorensetal} contains all places and transitions in any path that finishes in a place of the slicing criterion. In contrast, Figure~\ref{fig:llorensetalimproved} only contains the subnet that contains the shortest increasing firing sequence.
	
	Clearly, Algorithm~\ref{alg:Slicing-PN2} is more aggressive than Algorithm~\ref{alg:Slicing-PN}, and it has an interesting property related to the size of the slices:
	Algorithm~\ref{alg:Slicing-PN2} computes slices that are always subnets of those slices obtained by using Algorithm~\ref{alg:Slicing-PN}.
	
	Algorithm~\ref{alg:Slicing-PN} has an important property: maximality (stated and proven in Theorem~\ref{theo_completeness}).
	It preserves all places and transitions that can contribute tokens to the slicing criterion. This means that 
	it is especially useful for debugging. If we find a place with a wrong number of tokens, then the maximal slice computed w.r.t. that place necessarily contains the cause of the error. This notion of slice is equivalent to the original notion proposed by Mark Weiser \cite{Wei84} in the context of programming languages.
	
	Algorithm~\ref{alg:Slicing-PN2}, however, is not useful for debugging, because it could slice the cause of the error. Algorithm~\ref{alg:Slicing-PN2} is useful for Petri net comprehension and, particularly, for Petri net specialization. For instance, if given a marking of a Petri net, we want to extract a component (e.g., for reuse) that contributes tokens to a place, then the minimal slice computed by Algorithm~\ref{alg:Slicing-PN2} is exactly that component. The minimality of the slices computed by this algorithm is stated and proven in Theorem~\ref{Theo_Minimality}.
	
	\subsection{Petri net slicing algorithm 1: maximal contributing slice}
	
	\begin{algorithm*}\caption{Dynamic slicing of a marked
			Petri net: Maximal contributing slice}\label{alg:Slicing-PN}
		\algsetup{linenosize=\normalsize}
		\begin{algorithmic}[0]
			
			\REQUIRE A Petri net $\cN =(P,T,F)$ and a slicing criterion $\tuple{M_0,Q}$ for $\cN$\\
			\ENSURE It it exists, the maximal contributing slice $\cN'$ of $\cN$ with respect to $\tuple{M_0,Q}$\\[2ex]
			
			First, we compute a
			\emph{backward slice} $(P_b,T_b,F_b)$.  This is
			obtained from
			\[
			\hspace{3ex}  (P_b,T_b)=\bslice_\cN(Q,\{\:\}) \mbox{ with } F_b = F|_{(P_b,T_b)}
			\]
			Function
			$\bslice_\cN$ is defined as follows:\\[2ex]
			$ \bslice_\cN(W,W_{done}) =$
			$ \left\{ \begin{array}{ll}
			(\{\:\},\{\:\})  \hspace{40ex}\mbox{if}~ W=\{\:\} \\
			(\{p\}~\cup~{}^\bullet T_{in}, T_{in}) \oplus \bslice_\cN((W \cup {}^\bullet T_{in})\:\backslash W'_{done},W'_{done})
			\hspace{2ex} \mbox{if}~W\neq \{\:\}, \\\hspace{32ex}\mbox{where}~T_{in}={}^\bullet p
			,\\\hspace{32ex} W'_{done} = W_{done}\cup\{p\}~\mbox{for some } p \in W,\\ 
			\hspace{32ex}\mbox{and}~  (A,B) \oplus (A',B') = (A \cup A', B \cup B')\\
			\end{array} \right.
			$\\[2ex]
			
			Now, we compute a \emph{forward slice} $\cN_f = (P_f,T_f,F_f)$ from
			\[
			\hspace{3ex} (P_f,T_f) = \fslice_\cN(sc,\{\:\},\{t\in T_b\mid M_0\fire{t} \}) ~~\mbox{with } F_f = F_b|_{(P_f,T_f)}
			\]
			where the slicing criterion $sc$ of the forward slice is defined as: $\{p\in P_b\mid M_0(p)>0\}$,\\
			and function $\fslice_\cN$ is defined as follows:\\[2ex]
			
			$\fslice_\cN(W,R,V) =$
			$\left\{ \begin{array}{ll}
			(W, R)  & \hspace{1ex}\mbox{if}~ V=\{\:\} \\
			\fslice_\cN(W\cup V {}^\bullet,R\cup V,V')
			& \hspace{1ex}\mbox{if}~V\neq \{\:\},\\ & ~\mbox{where}~
			V' = \{ t \in T_b\:\backslash (R\cup V)
			\mid {}^\bullet t \subseteq W\cup V {}^\bullet \}\\ 
			\end{array} \right.
			$\\[2ex]

\IF {$\cN_f = (\emptyset, \emptyset, \emptyset)$}
        \RETURN ``no slice exists''
		\ELSE

		\STATE	The final slice $\cN' = (P',T',F|_{(P',T')})$ is composed of all places and transitions in the computed forward slice for which there exists a path to some place in the slicing criterion, such that:\\[1ex]

            \hspace{3ex}$P' = \{p~|~p \in P_f~\wedge~suc(p) \cap Q \neq \emptyset\}$\\[2ex]

            \hspace{3ex}$T' = \{t~|~t \in T_f~\wedge~suc(t) \cap Q \neq \emptyset\}$\\[1ex]

             where function $suc$ is the standard graph theory successor function.\\[1ex]

              \RETURN $\cN'$

\ENDIF\\[1ex]

		\end{algorithmic}
	\end{algorithm*}

	Algorithm~\ref{alg:Slicing-PN} describes our method to extract a
	dynamic slice from a Petri net.
	Intuitively speaking,
	Algorithm~\ref{alg:Slicing-PN} constructs the slice of a Petri net
	$\cN=(P,T,F)$ for a set of places $Q\subseteq P$ as follows. The key idea
	is to capture all possible token flows relevant for places in $Q$. For
	this purpose,
	\begin{itemize}
		\item we first compute the possible paths that lead to the slicing
		criterion, producing a backward slice;
		\item then from the backward slice, we compute the paths that may be followed by the
		tokens of the initial marking (those tokens that remain in the backward slice).
	\end{itemize}
	This can be done by taking into account that (i) the marking of a
	place $p$ depends on its input and output transitions, (ii) a
	transition may only be fired if it is enabled, and (iii) the
	enabling of a transition depends on the marking of its input places.
	The algorithm is divided into three steps:
	\begin{itemize}
		\item The first step is a backward slicing method (which is similar
		to the \emph{basic slicing algorithm} of \cite{Rak07}) that obtains
		a slice $\cN_b = (P_b,T_b,F_b)$ defined as the subnet of $\cN$ that
		includes all input places of all transitions transitively connected to any place
		$p$ in $P$, starting with $Q \subseteq P$.
		\begin{itemize}
			\item The core of this method is the auxiliary function $\bslice_\cN$, which is initially called with the set of places $Q$ of the slicing criterion together with an empty set of places.
			\item For a particular non-empty set of places $W$ and a particular place $p \in W$, function $\bslice_\cN$ returns the transitions $T$ in ${}^\bullet p$ and the input places of these transitions ${}^\bullet T$. Then, function $\bslice_\cN$ moves backwards (i) adding the place $p$ to the set $W_{done}$, (ii) adding ${}^\bullet T$ to $W$ and (iii) removing from $W$ the updated set $W_{done}$, until the set $W$ becomes empty.
		\end{itemize}
		\item The second step is a forward slicing method that obtains a slice
		$\cN_f = (P_f,T_f,F_f)$ defined as the subnet of $\cN_b$ that includes all places and transitions reachable from
		all transitions in $T_b$ initially enabled in $M_0$ as well as those
		places $p \in P_b$ such that $M_0(p)>0$.
		\begin{itemize}
			\item We define an auxiliary function $\fslice_\cN$, which is initially called with the places in $P_b$ that are marked at $M_0$, an empty set of transitions and the transitions in $T_b$ enabled  in $M_0$.
			\item For a particular set of places $W$, a particular set of transitions $R$ and a particular non-empty set of transitions $V$, function $\fslice_\cN$ moves forwards adding the places in $V^\bullet$ to $W$, adding the transitions in $V$ to $R$ and replacing the set of transitions $V$ by a new set $V'$ in which the transitions are included that are not in $R \cup V$ and whose input places are in $W \cup V^\bullet$.
			\item Finally, when $V$ is empty, function $\fslice_\cN$ returns the accumulated set of places and transitions $W \cup R$.
		\end{itemize}
		\item  The final dynamic slice is composed of all places and transitions in the computed forward slice for which there exists a path to some place in the slicing criterion.
	\end{itemize}

\begin{example}
	The different phases of Algorithm~\ref{alg:Slicing-PN} are depicted in Figure~\ref{fig:sliceAlg1}. The original Petri net only has one place ($p_3$, in grey) in the slicing criterion. This Petri net is sliced in the first phase with function $\bslice_\cN$, producing a backward slice. Then, function $\fslice_\cN$ is used to compute a forward slice. Finally, those parts not connected to the slicing criterion are removed to produce the final slice.
\end{example}

\begin{figure}[h!]
\vspace*{-3mm}
\centering
\includegraphics[scale=0.45]{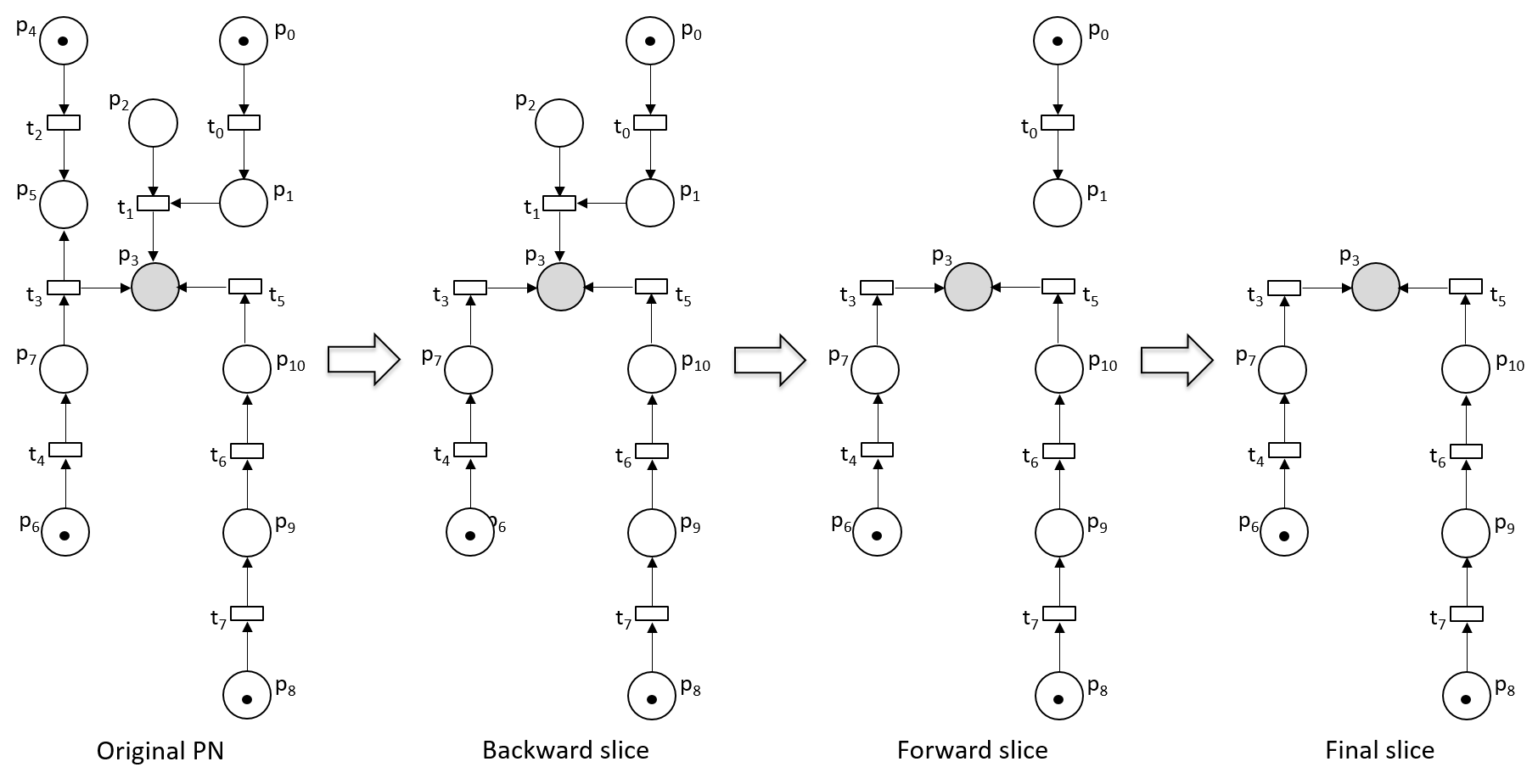}\vspace*{-1mm}
\caption{Phases of Algorithm 1.} \label{fig:sliceAlg1}\vspace*{-2mm}
\end{figure}

	This algorithm improves the formalization and precision of Algorithm~1 in \cite{Llo08}. There are two important differences between this algorithm and Algorithm~1 in \cite{Llo08}: (i) The net where the forward traversal is done. In \cite{Llo08} the forward slice is computed w.r.t. the initial net. In contrast, the improved algorithm computes the forward slice w.r.t. the backward slice produced. The forward slicing phase is, therefore, more efficient in the new algorithm. (ii) The new algorithm implements a third phase after the forward slicing phase. This phase filters out those subnets of the forward slice that cannot contribute tokens to some place of the slicing criterion; hence, improving the precision of the original algorithm.
	
	The cost of both algorithms is bounded by the number of transitions $T$ of the original Petri net: each transition is traversed at most twice in Algorithm~1 in \cite{Llo08} and at most three times in the new algorithm. Therefore, the asymptotic cost is $\cO(2T)$ and $\cO(3T)$, respectively. 
However, the new algorithm is monotonically more efficient because many places and transitions that are not reachable from the slicing criterion should not be processed (and they are processed in \cite{Llo08}); thus, the cost of the new algorithm is, in general, smaller than the algorithm in \cite{Llo08} (see empirical evaluation in Section~\ref{Sec-Evaluation}).
	
	The following result states that the net produced by Algorithm \ref{alg:Slicing-PN} is a slice of the input Petri net.
	
	\begin{theorem}[Soundness] \label{Theorem1}
		Let $\cN$ be a Petri net and $\tuple{M_0,Q}$ be a slicing criterion
		for $\cN$. The dynamic slice $\cN'$ computed in Algorithm~\ref{alg:Slicing-PN}
		is a valid slice according to Definition~\ref{slice1-def} or if the slice does not exist, then the algorithm returns ``no slice exists''.
	\end{theorem}
	
	In order to prove this result, we need first to prove the following lemma:

	\begin{lemma}\label{lemma1}
		Let $\cN$ be a Petri net and $\tuple{M_0,Q}$ be a slicing criterion for $\cN$.
		Let $\sigma$ be an increasing firing sequence of $\cN$ w.r.t. $\tuple{M_0,Q}$.
		Let $\cN_b=(P_b,T_b,F_b)$ be the backward dynamic slice computed by function $\bslice_\cN(Q, \{ \})$ in Algorithm \ref{alg:Slicing-PN}.
		There exists an increasing firing sequence $\sigma'$ in $\cN_b$ w.r.t. $\tuple{M_0,Q}$ such that $\sigma'$ is a subsequence of $\sigma$.
	\end{lemma}
	
	\begin{proof}
		Because $\sigma$ is an increasing firing sequence, then, by Definition~\ref{increasing-sequence-def}, we know that the marking of some place $p$ of the slicing criterion must be increased by the firing of the last transition in $\sigma$.
		Let us consider the set of transitions in $\sigma$ that must be fired to increase the marking of $p$.

		First, function $\bslice_\cN(Q, \{ \})$ includes all places of the slicing criterion in the slice, therefore, $p$ is included in the slice.
		Moreover, because it recursively takes all the incoming transitions together with their input places from the slicing criterion, then all their incoming places and transitions that must be fired to increase the marking of $p$ also belong to $\bslice_\cN(Q, \{ \})$.
		Therefore, there exists an increasing firing sequence $\sigma'$ in $\cN_b$ w.r.t. $\tuple{M_0,Q}$ such that $\sigma'$ is a subsequence of $\sigma$.
	\end{proof}
	
	Now, we prove Theorem~\ref{Theorem1}.
	
	\begin{proof} 
		Firstly, we need to prove the first condition of Definition~\ref{slice1-def}, i.e. $\cN'=(P',T',F')$ is a subnet of $\cN=(P,T,F)$ according to Definition~\ref{subnet}, such that $\not \exists p \in$ $^\bullet t, t \in T~|~t \in T' \wedge p \notin P'$. By construction, the net $\cN_b = (P_b,T_b,F_b)$ computed by function $\bslice_\cN$ is a subnet of $\cN$ because $P_b \subseteq P$, $T_b \subseteq T$, and $F_b = F|_{(P_b,T_b)}$. And, the net $\cN_f =(P_f,T_f,F_f)$ computed by function $\fslice_\cN$ is a subnet of $\cN_b$ because $P_f \subseteq P_b$, $T_f \subseteq T_b$, and $F_f = F_b|_{(P_f,T_f)}$. The final slice $\cN'$ is constructed by removing places and transitions from $\cN_f$. Therefore, $\cN'$ is a subnet of $\cN_f$ and thus a subnet of $\cN$. Moreover, function $\bslice_\cN$ always includes in the slice the input places of all included transitions. These places are kept by $\fslice_\cN$ and also by the final postprocess, thus $\not \exists p \in$ $^\bullet t, t \in T~|~t \in T' \wedge p \notin P'$.

\medskip
		Now, we prove the second condition of Definition~\ref{slice1-def}.
		If there does not exist an increasing firing sequence $\sigma$ in $\cN$ w.r.t. $\tuple{M_0,Q}$, then no slice exists. This is considered by Algorithm ~\ref{alg:Slicing-PN} when $\cN_f = (\emptyset, \emptyset, \emptyset)$.
		If an increasing firing sequence $\sigma$ exists, then the claim holds because a subsequence of $\sigma$, $\sigma_b$, exists in the backward slice of $\cN$ according to Lemma~\ref{lemma1}, and it also exists in the forward slice because function $\fslice_\cN$ always produces a slice where there exists an increasing firing subsequence $\sigma'$ of $\sigma_b$. This can be shown in an analogous way to the case of function $\bslice_\cN$ but in the opposite direction:

\medskip		
		First, because $\sigma_b$ is an increasing firing sequence, then we know that the marking of some place $p$ of the slicing criterion must be increased by the firing of the last transition in $\sigma_b$.
		Let us consider the set of transitions in $\sigma_b$ that must be fired to increase the marking of $p$.
		
		Function $\fslice_\cN$ recursively takes all the outgoing transitions together with their output places from those places in $P_b$ marked in $M_0$ as well as those transitions in $T_b$ enabled in $M_0$.
		So all transitions reachable from them are included in the slice, thus, all transitions in $\sigma_b$ that must be fired to increase the marking of $p$ are included in the forward slice.
		Hence, there exists an increasing firing sequence $\sigma'$ in $\cN_f$ w.r.t. $\tuple{M_0,Q}$ such that $\sigma'$ is a subsequence of $\sigma_b$. And, trivially, $\sigma'$ is a subsequence of $\sigma$.
Finally, the last step of the algorithm filters $\cN_f$ by removing places and transitions for which there does not exist a path to the slicing criterion. Therefore, the removed places and transitions cannot participate in the increasing firing sequence $\sigma'$.
		Hence, we can conclude that the resultant slice of Algorithm~\ref{alg:Slicing-PN} is always a valid slice according to Definition~\ref{slice1-def}.
	\end{proof}

	The following result states that the slice computed with Algorithm \ref{alg:Slicing-PN} is also a maximal slice (see Definition~\ref{maximalslice-def}).
	
	\begin{theorem}[Maximality]
		\label{theo_completeness}
		Let $\cN$ be a Petri net and $\tuple{M_0,Q}$ be a slicing criterion
		for $\cN$. The dynamic slice $\cN'$ computed in Algorithm \ref{alg:Slicing-PN}
		is a maximal slice.
	\end{theorem}
	
	\begin{proof}
First, by Theorem \ref{Theorem1}, we know that the dynamic slice $\cN'$ computed in Algorithm~\ref{alg:Slicing-PN}
		is a valid slice according to Definition~\ref{slice1-def} or if the slice does not exist, then the algorithm returns ``no slice exists''. Hence, we only need to prove the second bullet of Definition~\ref{maximalslice-def}.
		The proof follows easily by contradiction assuming that the slice produced by Algorithm \ref{alg:Slicing-PN} is not maximal.
		Because the slice is not maximal, then, by Definition~\ref{maximalslice-def}, there must exist an increasing firing sequence $\sigma$ of $\cN$ w.r.t. $\tuple{M_0,Q}$ for which there does not exist an increasing firing sequence $\sigma'$ in $\cN'$ w.r.t. $\tuple{M_0,Q}$ such that $\sigma'$ is a subsequence of $\sigma$.

\medskip		
		We prove that $\sigma'$ exists in $\cN'$, which is a contradiction.
		
		First, there exists a sequence $\sigma_b$ which is a subsequence of $\sigma$. This holds because
		function $\bslice_\cN(Q, \{ \})$ includes all places of the slicing criterion in the slice, and all the incoming places and transitions that must be fired to increase the marking of any place in the slicing criterion also belong to $(P_b,T_b)=\bslice_\cN(Q, \{ \})$.
		Thus, $\sigma_b$ is a subsequence of $\sigma$.
		
		Second, $\sigma'$ is a subsequence of $\sigma_b$. This holds because
		function $\fslice_\cN$ recursively takes all the outgoing transitions together with their output places from those places in $P_b$ marked in $M_0$ as well as those transitions in $T_b$ enabled in $M_0$.
		So all transitions reachable from them are included in the slice, thus, all transitions in $\sigma_b$ that must be fired to increase the marking of any place in the slicing criterion are included in the forward slice. Moreover, none of the places and transitions removed for the forward slice can contribute tokens to the slicing criterion. Therefore, their removal can make $\sigma'$ smaller, but in any case $\sigma'$ will still be an increasing firing sequence and a subsequence of $\sigma$.

\medskip		
		We can conclude that the resultant slice of Algorithm~\ref{alg:Slicing-PN} is always a maximal slice according to Definition~\ref{maximalslice-def}.
	\end{proof}

	\subsection{Petri net slicing algorithm 2: minimal contributing slice}\label{sec:alg2}

	\begin{algorithm*}\caption{Dynamic slicing of a marked Petri net: Minimal contributing slice}\label{alg:Slicing-PN2}
	\algsetup{linenosize=\normalsize}
	\begin{algorithmic}[0]
			
\REQUIRE A Petri net $\cN =(P,T,F)$ and a slicing criterion $\tuple{M_0,Q}$ for $\cN$\\
\ENSURE If it exists, the minimal contributing slice $\cN'$ of $\cN$ with respect to $\tuple{M_0,Q}$\\[2ex]
		
		First, we compute a \emph{backward slice}. This is obtained from $\bslice_\cN(Q,\{\:\})$:\\[1ex]
		\STATE \hspace{3ex}$\cN_b = (P_b,T_b,F|_{(P_b,T_b)}) = \bslice_\cN(Q,\{\:\});$\\[1ex]
		Second, we remove those branches of the backward slice that cannot contribute tokens to the slicing criterion:\\[1ex]
		\STATE \hspace{3ex}$\mathit{filteredBackwardSlice} = \filter_\cN(\mathit{\cN_b});$\\[1ex]
		Third, we compute a \emph{forward slice} of the backward slice:\\[1ex]
		\STATE \hspace{3ex}$\cN_f = (P_f,T_f,F|_{(P_f,T_f)}) = \fslice(\mathit{filteredBackwardSlice}, \tuple{M_0|_{P_b},Q});$\\[1ex]
		
		
		%
		
		We return the forward slice as the final dynamic slice:\\[1ex]
		
        \IF {$\cN_f = (\emptyset, \emptyset, \emptyset)$}
        \RETURN ``no slice exists''
		\ELSE \RETURN $\cN' = \cN_f$
		\ENDIF\\[1ex]

		The auxiliary function $\fslice$ is defined in Algorithm~\ref{alg-forward}, and the auxiliary functions $\bslice_\cN$ and $\filter_\cN$ are defined as follows:\\[1ex]
		$ \bslice_\cN(W,W_{done}) =$
		$ \left\{ \begin{array}{ll}
		(\{\:\},\{\:\})  \hspace{40ex}\mbox{if}~ W=\{\:\} \\
		(\{p\}~\cup~{}^\bullet T_{in}, T_{in}) \oplus \bslice_\cN((W \cup {}^\bullet T_{in})\:\backslash W'_{done},W'_{done})
		\hspace{2ex} \mbox{if}~W\neq \{\:\}, \\\hspace{32ex}\mbox{where}~T_{in}={}^\bullet p
		,\\\hspace{32ex} W'_{done} = W_{done}\cup\{p\}~\mbox{for some } p \in W,\\ 
		\hspace{32ex}\mbox{and}~  (A,B) \oplus (A',B') = (A \cup A', B \cup B')\\
		\end{array} \right.
		$\\[2ex]
		
		$\filter_\cN((P,T,F)) = $\\
        \WHILE {($\exists$ $p \in P$ such that $M_0(p) ==0 \wedge {}^\bullet p == \emptyset$)}
		 	\STATE $P = P \backslash \{p\}$
		    \STATE $T = T \backslash p^\bullet$					
		 \ENDWHILE
		\RETURN $(P,T,F|_{(P,T)})$
		
	\end{algorithmic}
\end{algorithm*}

To properly formalize Algorithm~\ref{alg:Slicing-PN2}, we introduce the definition of \emph{forward slicing tree}.

\begin{definition}\label{def:fst}Let $\cN$ be a Petri net and let $\tuple{M_0,Q}$ be a slicing criterion for $\cN$.
A \emph{forward slicing tree} $\cT = (V, E)$ of the marked Petri net $(\cN,M_0)$ is a reachability tree where edges represent weighted transitions between markings: $E \subseteq (V \times (T,\Nat) \times V)$ such that $(M, (t,n), M') \in E$ if $M\fire{t}M'$, with $M, M' \in V$, $t \in T$ and $n \in \Nat$ is the weight assigned to transition $t$ fired on marking $M$ as follows: $w(M,t) = l(M) + l(t)$
where $l(M)$ is the length of the path in $\cT$ from the root $M_0$ to $M$ defined as the number of edges from $M_0$ to $M$; and
$l(t)$ is the length of the minimum path in $\cN$ 
from $t$ to some $q\in Q$, defined as the minimum number of transitions from $t$ (included) to $q$.
\end{definition}

\begin{algorithm*}\caption{$\fslice(\cN, \tuple{M_0,Q})$}\label{alg:Slicing-PN3}
\algsetup{linenosize=\normalsize}
\label{alg-forward}
\begin{algorithmic}[1]
\REQUIRE A Petri net $\cN =(P,T,F)$ and a slicing criterion $\tuple{M_0,Q}$ for $\cN$\\
\ENSURE The forward slice $\cN'$ of $\cN$ with respect to $\tuple{M_0,Q}$\\
\STATE \textbf{Initialization}: $\cT = \tuple{M_0, \emptyset}$, where $\cT = \tuple{V, E}$ is a forward slicing tree.
\STATE \textbf{Begin}
\IF {$\not\exists M_f \in R(\cN, M_0)$ such that $M_0(p) < M_f(p)$ with $p \in Q$}
   \RETURN $(\emptyset, \emptyset, \emptyset)$ (i.e., no slice exists)
\ENDIF
\STATE $enabledTransitions = \{(v,t) \mid v\in V, t\in T, v\fire{t}\}$.
    \WHILE {$enabledTransitions \neq \emptyset$}
        \STATE Choose $(v,t) \in enabledTransitions$ such that\\ $\nexists (v',t') \in enabledTransitions$~where $w(v',t') < w(v,t)$.
        \STATE $enabledTransitions = enabledTransitions \backslash \{(v,t)\}$
        \STATE $V = V \cup \{v'\}$ where $v\fire{t}v'$
        \STATE $E = E \cup \{(v,(t,w(t)),v')\}$ 
        \IF {$\not\exists v'' \in V$ such that $v''$ is an ancestor of $v'$ and $v' == v''$}
            \STATE $enabledTransitions = enabledTransitions \cup \{(v',t') \mid v'\in V, t'\in T, v'\fire{t'}\}$.
        \ENDIF
        \IF {$v'(p)>v(p)$, for some $p \in Q$}
            \STATE $T' = \{t \mid t \in \mathit{RTPath}(M_0,v')\}$
            \STATE $P' = \{p\}$
            \FORALL {$t\in T'$}
                   \STATE $P' = P' \cup {}^\bullet t$
            \ENDFOR
            \RETURN $\cN' = (P',T',F|_{(P',T')})$ (i.e., the forward slice)
        \ENDIF

    \ENDWHILE

\STATE \textbf{End}
\end{algorithmic}
\end{algorithm*}

	This algorithm tries to further reduce the size of the slice produced by identifying the forward slice with the shortest increasing firing sequence. Thus, it selects the smallest set of transitions that can contribute tokens to the slicing criterion, together with their input places.
	
	The main difference between Algorithm~\ref{alg:Slicing-PN} and Algorithm~\ref{alg:Slicing-PN2} is the forward slicing phase.
	The basic idea is the following: Once the backward slice has been computed, the algorithm builds a forward slicing tree from the initial marking. A transition is added to the tree if and only if this transition has the minimum weight $w(M,t)= l(M) + l(t)$ which means that this transition is the one that can be fired after less transitions from the initial marking ($l(M)$) and at the same time it needs to fire less transitions to reach the slicing criterion ($l(t)$).

\eject	
	The algorithm performs the following phases:
	\begin{itemize}\itemsep -1pt
		\item \underline{computing the backward slice with function $\bslice_\cN$:} iteratively collecting all the incoming transitions together with their input places from the slicing criterion;
		\item \underline{discarding useless branches in the backward slice with function $\filter_\cN$:} those that do not contain any token nor any enabled transition are discarded;
		\item \underline{computing the minimal forward slice in the forward slicing tree with function $\fslice_\cN$:} This function, implemented by Algorithm~\ref{alg:Slicing-PN3}, iteratively expands the tree with the transition with the minimum weight until an increasing firing sequence is found. This increasing firing sequence is necessarily the shortest one, so that the algorithm often only builds a portion of the forward slicing tree. The final output is the net formed from the places and transitions needed to fire the sequence with the minimum weight.
		
		Note that the forward slicing tree could be infinite. Therefore, the algorithm could enter into an infinite search process if a branch is infinitely explored because no slice exist.
		To avoid this situation, the algorithm determines first whether a slice exist (lines 3-5). This can be determined in finite time and space with a coverability tree. If the slice doesn't exist then the algorithm returns $(\emptyset, \emptyset, \emptyset)$. Therefore, Algorithm~\ref{alg:Slicing-PN3} always terminates.
	\end{itemize}

\begin{example}
	The different phases of Algorithm~\ref{alg:Slicing-PN2} are depicted in Figure~\ref{fig:sliceAlg2}. The original Petri net only has one place ($p_3$, in grey) in the slicing criterion. This Petri net is sliced in the first phase with function $\bslice_\cN$, producing a backward slice. Then, function $\filter_\cN$ discards useless branches in the backward slice. Finally, function $\fslice_\cN$ (Algorithm~\ref{alg:Slicing-PN3}) builds the forward slicing tree shown in Figure~\ref{fig:forwardSlicingTree} to compute the minimal forward slice shown at the right-hand side of Figure~\ref{fig:sliceAlg2}.

\begin{figure}[h!]
\centering
\includegraphics[scale=0.47]{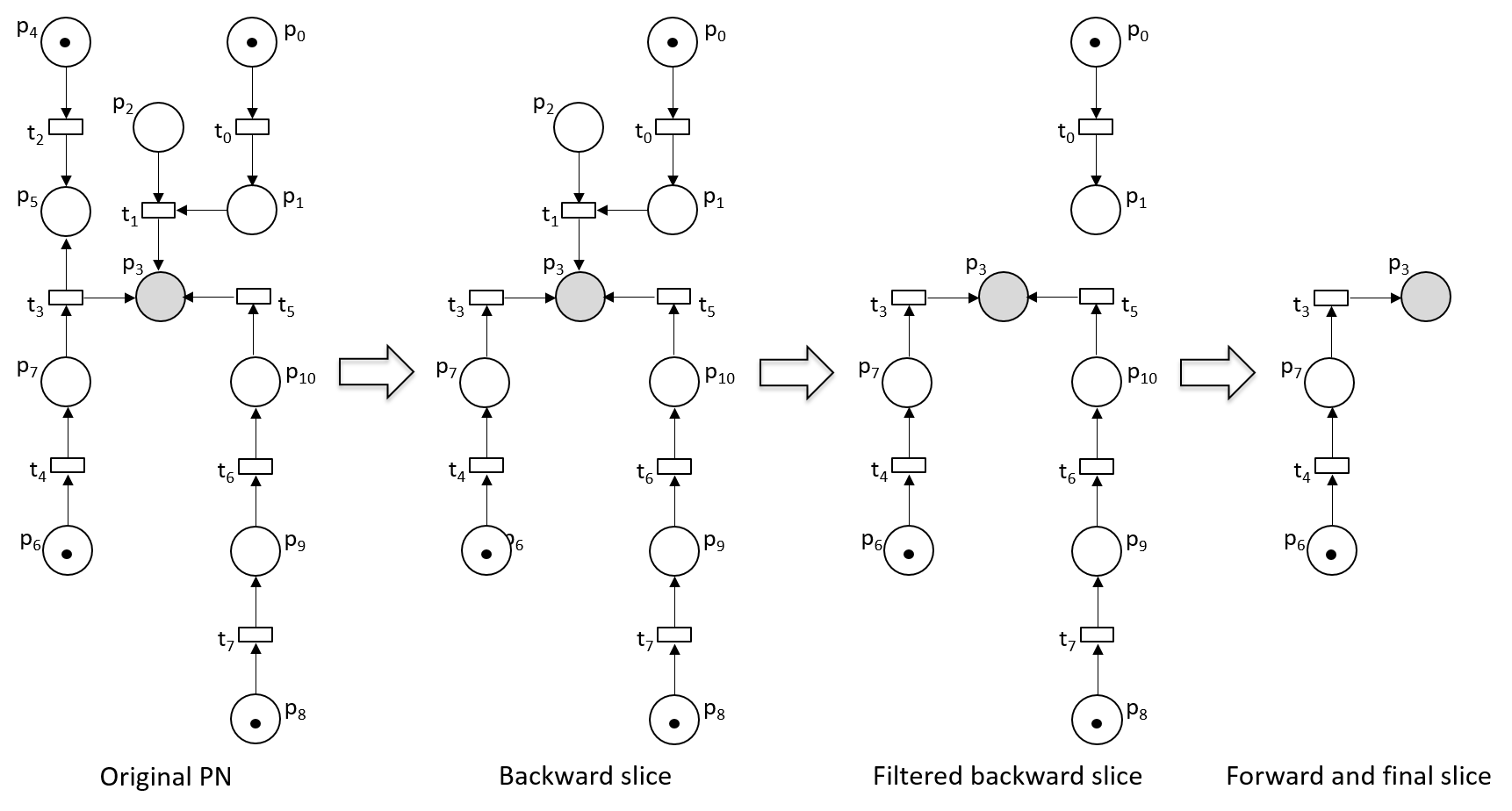}\vspace*{-1mm}
\caption{Phases of Algorithm 2.} \label{fig:sliceAlg2}
\end{figure}

\begin{figure}[h!]
\centering
\includegraphics[scale=0.47]{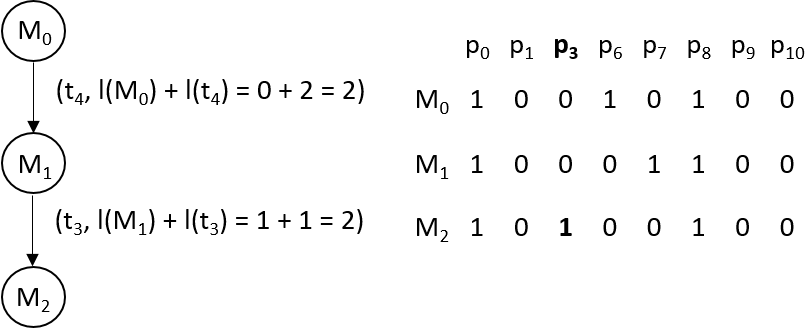}\vspace*{-1mm}
\caption{Forward slicing tree.} \label{fig:forwardSlicingTree}
\end{figure}
\end{example}
	
	The following result states that the net produced by Algorithm \ref{alg:Slicing-PN2} is a slice of the input Petri net.
	
	\begin{theorem}[Soundness]\label{Theorem3}
		Let $\cN$ be a Petri net and $\tuple{M_0,Q}$ be a slicing criterion
		for $\cN$. A dynamic slice $\cN'$ computed with Algorithm \ref{alg:Slicing-PN2}
		is a valid slice according to Definition~\ref{slice1-def} or if the slice does not exist, then the algorithm returns ``no slice exists''.
	\end{theorem}
	
	In order to prove this result, we need first to prove the following lemma:
	
	\begin{lemma}\label{lemma2}
		Let $\cN$ be a Petri net and $\tuple{M_0,Q}$ be a slicing criterion for $\cN$.
		Let $\cN_b$ be the backward dynamic slice computed by function $\bslice_\cN(Q, \{ \})$ in Algorithm \ref{alg:Slicing-PN2}. Let $\mathit{filteredBackwardSlice}$ be the filtered backward dynamic slice computed by function $\filter_\cN(\cN_b)$ in Algorithm \ref{alg:Slicing-PN2}.
		For each increasing firing sequence $\sigma$ of $\cN$ w.r.t. $\tuple{M_0,Q}$, there exists an increasing firing sequence $\sigma'$ w.r.t. $\tuple{M_0,Q}$ such that:
		\begin{itemize}
\itemsep=0.95pt
			\item $\sigma'$ exists in $\mathit{filteredBackwardSlice}$.
			\item $\sigma'$ is a subsequence of $\sigma$.
		\end{itemize}
	\end{lemma}
	
	\begin{proof}
		First, function $\bslice_\cN(Q, \{ \})$ iteratively collects, from the places in the slicing criterion, each incoming transition together with their input places to form an independent slice ($\cN_b$). Therefore, all possible paths that reach the slicing criterion belong to $\cN_b$, and, thus, a subsequence of all possible increasing firing sequences in $\cN$ belong to $\cN_b$. This is ensured because the tokens in the slice can traverse all the paths to the slicing criterion that can be traversed in $\cN$.

\medskip		
		Second, $\filter_\cN(\cN_b)$ only removes from $\cN_b$ those branches that do not contain any token or any source transition, and there does not exist any non-marked place without input transitions;
		i.e.
		$\{(P \cup T) \in \cN_b \mid
		((\forall ~p \in P ~.~ M(p)=0
		\wedge
		(\forall ~p' \in P, p'\in predecessor(p) ~.~ M(p')=0)
		\wedge
		(\forall ~t' \in T, t'\in predecessor(p) ~.~ {}^\bullet t' \neq \emptyset))
		\wedge
		(\forall ~t \in T ~.~ {}^\bullet t \neq \emptyset
		\wedge
		(\forall ~p'' \in P, p''\in predecessor(t) ~.~ M(p'')=0)
		\wedge
		(\forall ~t'' \in T, t''\in predecessor(t) ~.~ {}^\bullet t'' \neq \emptyset))
		\}
		$.
		This ensures that all transitions in $\cN_b$ that can be enabled at some point in $\cN$ can also be enabled in $\cN_b$.
		Therefore, for each place in $Q$ there exists a subsequence of all increasing firing sequences in $\cN$ in $\mathit{filteredBackwardSlice}$.
		
		As a consequence, for all $\sigma$ in $\cN$, there exits a $\sigma'$ in $\mathit{filteredBackwardSlice}$ such that $\sigma'$ is a subsequence of $\sigma$.
	\end{proof}
	
	Now, we prove Theorem~\ref{Theorem3}.
	
	\begin{proof}
		First, we prove the first condition of Definition~\ref{slice1-def}, i.e. $\cN'=(P',T',F')$ is a subnet of $\cN=(P,T,F)$ according to Definition~\ref{subnet}, such that $\not \exists p \in$ $^\bullet t, t \in T~|~t \in T' \wedge p \notin P'$. By construction, $\mathit{filteredBackwardSlice}$ is a subnet of $\cN$ because they only contain places and transitions of the original net. Similarly, $\cN_f$ is also a subnet of $\cN$. Finally, Algorithm ~\ref{alg:Slicing-PN3} in line 19 includes in the slice the input places of all included transitions. Therefore, $\not \exists p \in$ $^\bullet t, t \in T~|~t \in T' \wedge p \notin P'$.
		
		Second, we prove the second condition of Definition~\ref{slice1-def}.
		Let $\sigma$ be an increasing firing sequence in $\cN$. By Lemma~\ref{lemma2} we know that for each increasing firing sequence $\sigma$ in $\cN$, there exists an increasing firing sequence $\sigma_b$ in $\mathit{filteredBackwardSlice}$ such that $\sigma_b$ is a subsequence of $\sigma$.
		
		In addition, there exists an increasing firing sequence in $\mathit{filteredBackwardSlice}$ because function $\filter_\cN(\cN_b)$ removes all branches that do not participate in any firing sequence.
		
		Function $\fslice_\cN$ is an incremental algorithm that iteratively collects a transition to construct the forward slicing tree. It either computes a subnet that contains an increasing firing sequence or, if it does not exist, it returns the empty slice (denoted with the tuple $(\emptyset, \emptyset, \emptyset)$). Therefore, if a non-empty slice is computed it must contain an increasing firing sequence from $M_0$. The returned slice is composed of all places and transitions used in the increasing firing sequence. 		
		Finally, there must exist an increasing firing sequence $\sigma$ in $\cN$ such that $\sigma'$ is a subsequence of $\sigma$. This holds because the slice is a subnet of $\cN$ according to Definition~\ref{slice1-def}.

If the slice does not exist ($\cN_f = (\emptyset, \emptyset, \emptyset)$), the algorithm returns ``no slice exists''.
	\end{proof}
	
	The following result states that the slice computed with Algorithm \ref{alg:Slicing-PN2} is minimal (i.e. it needs to fire a minimum number of transitions to contribute a token to the slicing criterion).
	
	\begin{theorem}[Minimality]
		\label{Theo_Minimality}
		Let $\cN$ be a Petri net and $\tuple{M_0,Q}$ be a slicing criterion
		for $\cN$. If Algorithm~\ref{alg:Slicing-PN2} computes a dynamic slice $\cN'$, then $\cN'$
		is a minimal Petri net slice.
	\end{theorem}
	
	\begin{proof}
		The proof of this result follows easily by showing that there does not exist a valid slice that contains an increasing firing sequence that is shorter than the one contained in the dynamic slice obtained by Algorithm \ref{alg:Slicing-PN2}.
		
		First, by Lemma~\ref{lemma2}, we know that
				for each increasing firing sequence $\sigma$ of $\cN$ w.r.t. $\tuple{M_0,Q}$, there exists in $\mathit{filteredBackwardSlice}$ an increasing firing subsequence $\sigma'$ w.r.t. $\tuple{M_0,Q}$.
		In the proof of that Lemma, we have also shown that $\mathit{filteredBackwardSlice}$ contains all paths to the slicing criterion, which include all transitions that can participate in any increasing firing sequence. Therefore, the minimal slice must be a subnet of $\mathit{filteredBackwardSlice}$.

		It is important to remark that the algorithm uses a weighting system that explores first those nodes that could reach the minimal slice (avoiding the depth-first search and using instead a more or less breadth-first search). In the worst case, all branches are explored up to the minimal path size, being equivalent to a breadth-first search.
		
		Now, we can prove this Theorem by contradiction assuming that function $\fslice_\cN$ can compute a slice that is not minimal.
		
		(Hypothesis) We assume that the slice $\cN'$ cannot fire the shortest increasing firing sequence $\sigma''$ that can contribute tokens to the slicing criterion. In other words, we assume that the slice $\cN'$ can fire an increasing firing sequence $\sigma' = t'_1 \ldots t'_n$ in $\cN'$ w.r.t. $\tuple{M_0,Q}$; and there exists a slice $\cN''=(P'',T'',F'')$ of $\cN$ w.r.t.\ $\tuple{M_0,Q}$ with an increasing firing sequence $\sigma'' = t''_1 \ldots t''_m$ such that $m<n$.
		If the hypothesis is true, then there must be a node $(v',t')$ in the forward slicing tree (the one selected by the algorithm) whose weight $w'(v',t')$ is the minimum and whose path from the root is $\sigma'$.
		Moreover, the node $(v'',t'')$ whose path from the root is $\sigma''$ must have a weight $w''(v'',t'')$ such that $w'(v',t')<w''(v'',t'')$. Otherwise, the algorithm would have selected the node $w''(v'',t'')$.
		
		But this is a contradiction because $w$ is defined as follows: $w(M,t) = l(M) + l(t)$ (see Definition~\ref{def:fst}). And, thus,
		$l(t'_n)=l(t''_m)=1$ and $w(v'',t'')<w(v',t')$ because $m<n$. We reach the contradiction: $w'(v',t')<w''(v'',t'') \wedge w'(v',t')>w''(v'',t'')$. Therefore, $\cN''$ cannot exist.
		
	    Finally, according to the algorithm, the Petri net $\cN'$ only contains the places and transitions needed to fire the increasing firing sequence $\sigma'$ in $\cN'$ w.r.t. $\tuple{M_0,Q}$. Then, $\cN'$ is the minimal Petri net slice of $\cN$ w.r.t. $\tuple{M_0,Q}$.
	\end{proof}
	
	\subsection{Computing the smallest slice}\label{sec:minimal-slice2}

As shown in Example \ref{example-minimal-smallest}, the minimal slice is not necessarily the smallest slice. If we want to compute the smallest slice, we need to make some modifications to Algorithm \ref{alg:Slicing-PN3}. In this section we provide a definition of \emph{smallest} slice, which is a slice with the minimum number of different transitions needed to increase the number of tokens in the slicing criterion. We measure the size of a Petri net $\cN$, denoted $|\cN|$, with its number of transitions $|T|$.
	
	\begin{definition} \label{minimalslice-def2} Let $\cN$ be a Petri
		net and let $\tuple{M_0,Q}$ be a slicing criterion for $\cN$. Given
		a Petri net $\cN'=(P',T',F')$, we say that $\cN'$ is a \emph{smallest} slice of
		$\cN$ if the following conditions hold:
		\begin{itemize}
			\item the Petri net $\cN'$ is a slice of $\cN$ w.r.t.\ $\tuple{M_0,Q}$,
			\item there does not exist a slice $\cN''=(P'',T'',F'')$ of $\cN$ w.r.t.\ $\tuple{M_0,Q}$ such that $|T''|<|T'|$.
		\end{itemize}
	\end{definition}

As it happens with the minimal slice, the smallest slice of a Petri net is not necessarily unique. There can coexist many different smallest slices of a given Petri net.

To properly formalize Algorithm~\ref{alg:Slicing-PN3}, we slightly change the definition of forward slicing tree as follows.
\begin{definition}\label{def:fst2}Let $\cN$ be a Petri net and let $\tuple{M_0,Q}$ be a slicing criterion for $\cN$.
A \emph{forward slicing tree} $\cT = (V, E)$ of the marked Petri net $(\cN,M_0)$ is a reachability tree where edges represent weighted transitions between markings: $E \subseteq (V \times (T,\Nat) \times V)$ such that $(M, (t,n), M') \in E$ if $M\fire{t}M'$, with $M, M' \in V$, $t \in T$ and $n \in \Nat$ is the weight assigned to transition $t$ fired on marking $M$ as follows: $w(M,t) = l(M) + l(t)$
where $l(M)$ is the number of edges labeled with different transitions in the path in $\cT$ from the root $M_0$ to $M$; and $l(t)$ is the length of the minimum path in $\cN$ 
from $t$ to some $q\in Q$.
\end{definition}

The only difference between Definitions \ref{def:fst2} and \ref{def:fst} is how $l(M)$ is computed. Now, $l(M)$ represents the number of different transitions in the path in $\cT$ from $M_0$ to $M$ instead of all the transitions in this path. This ensures that we explore first the paths in $\cT$ whith the mininum number of different transitions, thus the final slice is the smallest one.

\medskip
In Algorithm~\ref{alg:Slicing-PN3}, we also need to make a change: line 5 must be replaced by the following\\

5:~~~~~Choose $(v,t) \in enabledTransitions$ such that\\
\indent$ ~~~~~~~~~\!\nexists (v',t') \in enabledTransitions$~where $(w(v',t') < w(v,t)) \vee (w(v',t') = w(v,t) \wedge l(t') < l(t))$\\

In the original algorithm, the path explored in $\cT$ was always the one with the minimum weight. In the case of a draw, any path was selected. In contrast, the above change in line 5 selects one specific path in the case of a draw: it selects the enabled transition that is closer to a place in the slicing criterion.

With these changes, in Example \ref{example-minimal-smallest}, the slice produced will be the smallest slice (instead of the minimal slice).

\subsection{Properties preserved and not preserved by the slices}\label{sec:alg_props}
	
	It is important to consider the properties\footnote{The definition of properties can be found at
       \url{https://github.com/tamarit/pn_suite/blob/master/doc/ glossary.pdf}} of the original net that are preserved in a maximal or minimal slice.
	Let us consider a Petri net $\cN$ and a slicing criterion $SC = \tuple{M_0,Q}$ for $\cN$. Let $\cN'$ be the slice computed with Algorithm \ref{alg:Slicing-PN} or with Algorithm \ref{alg:Slicing-PN2} for $\cN$ w.r.t. $SC$.
	
\medskip
	We have that $\cN'$ keeps the following properties of $\cN$ (i.e., if $\cN$ has property $p$, then $\cN'$ also has property $p$):
	\begin{description}
		\item behavioural: persistent, bounded, k-bounded (but k can be reduced), k-marking (but k can be reduced), and safe.
		\item structural: free-choice, restricted free-choice, asymmetric choice, pure, homogeneous, plain, conflict-free, output non-branching, t-net, and s-net.
	\end{description}

For the Petri net slice to stop preserving the above properties, it would be needed to add new places or transitions to the original net,
or change the flow relations, or change the initial marking. But, our algorithms do not change the structure nor the marking of the original net, they just remove places or transitions.  	

\medskip
	On the other hand, we have that a maximal slice $\cN'$ not necessarily keeps the following properties of $\cN$ (i.e., if $\cN$ has property $p$, then $\cN'$ could not preserve property $p$):
	\begin{description}
		\item behavioural: backwards-persistent.
		\item structural: isolated elements can disappear, siphons can disappear, and traps can disappear.
	\end{description}

We have that a minimal slice $\cN'$ not necessarily keeps the following properties of $\cN$ (i.e., if $\cN$ has property $p$, then $\cN'$ could not preserve property $p$):
	\begin{description}
		\item behavioural: strongly live, weakly live, and reversible.
		\item structural: strongly connected, isolated elements can disappear, siphons can disappear, and non-pure only simple side conditions.
	\end{description}
	
	For each of these properties, we have computed a counter-example (a Petri net where property $p$ holds and a slicing criterion for which the slice computed does not satisfy $p$) with our implementation and the property verifier of APT \cite{apt}. They can be found at \url{https://github.com/tamarit/pn_suite/blob/master/doc/properties.pdf}. Each counter-example (a Petri net and a slicing criterion) indicates what properties are not preserved.
	
	\section{Related work: other Petri net slicing algorithms}\label{sec_stateArt}
	
	
	In this section, we review the related work concerning Petri net slicing algorithms.
	We present a list of Petri net slicing algorithms (the most relevant ones) in the current state of the art \cite{Llo08,Rak12,Yu15}.
	
	\medskip
	\noindent\textbf{CTL$^{*}_{-x}$ Slicing} \cite{Rak12}: This approach is backwards and static. It produces slices that preserve a CTL$^{*}_{-x}$ property (that is CTL$^{*}$ but without next-time operator). In this approach, the slicing criterion is defined by a list of places (for example, those places referred to in a CTL$^{*}_{-x}$ property). The slice is computed from the places in the slicing criterion by (iteratively) collecting all the transitions that change the marking of a place, i.e. the incoming and outgoing non-reading transitions together with their related input places.
	The corresponding CTL$^{*}_{-x}$ algorithm identifies all those paths that can change (decrease or increase) tokens in the slicing criterion (any place), and it also identifies the paths that could disable or enable a transition in those paths. 
	
	\medskip
	\noindent\textbf{Safety Slicing} \cite{Rak12}: This approach is backwards and static. It extracts a subnet that preserves stutter-invariant linear-time safety properties. The slicing criteria in this approach is composed of sets of places $\mathit{Q}$. The corresponding algorithm collects all non-reading transitions connected to $\mathit{Q}$ and all their input places. Then, iteratively, it collects (only) those transitions (and their input places) that could increase the tokens in the sliced net.
	%
	This algorithm identifies all the paths that can increase the tokens in the slicing criterion ensuring that in the resulting slice: (i) all the places contain the same or more tokens as in the original net; and (ii) the places in the slicing criterion keep the same number of tokens.
	
	\medskip
	\noindent\textbf{Slicing by Llorens et al.} \cite{Llo08}: This algorithm was the first approach for dynamic slicing. It is backwards, but it uses a forwards and a backwards traversal of the Petri net. The slicing criterion in this approach is a pair $\langle M_{0}, Q\rangle$, where $M_{0}$ is the initial marking and $Q$ is a set of places.
	This algorithm computes two slices: (i) a backward slice is computed collecting all incoming transitions plus their relating input places. (ii) a forward slice is computed as follows: first, it takes all places marked in $M_{0}$ and all transitions initially enabled in $M_{0}$. Then, the outgoing places and the transitions whose input places are in the slice are iteratively included in the slice. Finally, the algorithm computes the intersection of the forward and backward slices to produce the final slice.
	This algorithm identifies all the paths that, from the initial marking, can contribute tokens to any place in the slicing criterion.

	\medskip
	\noindent\textbf{Slicing by Yu et al.} \cite{Yu15}: This approach is backwards and dynamic. The underlying data structure used to produce slices is the Structural Dependency Graph (SDG). Two algorithms are used: (i) The first algorithm defines the slicing criterion as a set of places $Q$, and it builds a $\mathit{SDG}(N)$ by traversing the SDG backwards from $Q$.
	(ii) The second algorithm starts from $\mathit{SDG}(N)$. It defines the slicing criterion as $\langle M_{0}, Q\rangle$. The dynamic slice is the subnet of $N$ that can dynamically influence the slicing criterion from the initial marking $M_{0}$. It could be possible that the initial marking $M_{0}$ cannot affect the slicing criterion. In such a case, the slice is empty. This means that there does not exist any subnet that can introduce tokens to the slicing criterion.
	This algorithm identifies one single path (there could be others) that, from the initial marking, can augment the number of tokens in at least one place of the slicing criterion.
	
	%
	%
	
	
	\begin{example}
		We want to extract a slice from the Petri net in Figure~\ref{fig:initialPNandRakowCTL} with respect to the slicing criterion $\{p_6,p_9\}$.
		The slice produced by the CTL$^{*}_{-x}$ algorithm is exactly the whole original net (Figure~\ref{fig:initialPNandRakowCTL}).
		The slice produced by the safety algorithm by Rakow is depicted in Figure~\ref{fig:rakowSafety}.
		The slice produced by the algorithm by Llorens et al., and by Algorithm~\ref{alg:Slicing-PN}, is the one in Figure~\ref{fig:llorensetal}.
		The slice produced by Algorithm~\ref{alg:Slicing-PN2} is shown in Figure~\ref{fig:llorensetalimproved}.
		Finally, the slice produced by Yu et al.'s algorithm is depicted in Figure~\ref{fig:yuetal}.
	\end{example}

	\section{Implementation: \texttt{pn$\_$slicer}}\label{PN-Slicer}
	
	
	This section describes a tool named \texttt{pn$\_$slicer} that was originally proposed in \cite{Llo17}. It is a system that implements the most important slicing algorithms for Petri nets and we have extended it with the Algorithms \ref{alg:Slicing-PN} and \ref{alg:Slicing-PN2} presented in this paper. 
	Currently, \texttt{pn$\_$slicer} implements five slicing algorithms but it is ready to easily integrate any slicing algorithm so that researchers can plug in their algorithms into \texttt{pn$\_$slicer}. Therefore, this system can be seen as a workbench with facilities for slicing.
	\texttt{pn$\_$slicer} is particularly useful for the optimisation and analysis of Petri nets. For instance, it allows us to slice a Petri net with all existing algorithms and produce reports about their sizes and about what properties have been preserved in each slice.
	
	Due to efficiency reasons, \texttt{pn$\_$slicer} does not implement lines 3-5 in Algorithm \ref{alg:Slicing-PN3}. The condition used to stop is a timeout instead. In most cases, this makes our implementation to return the minimal slice without the need to build the coverability tree.
	
	The evaluation of properties (liveness, bounded, etc.) is performed by APT \cite{apt} through the use of communication interfaces that are transparent for the user. In this way, \texttt{pn$\_$slicer} internally calls APT analyses to decide what properties are kept or lost in the produced slices.
	
	\texttt{pn$\_$slicer} is open-source and free. It can be downloaded from \url{https://github.com/tamarit/pn_suite}.
	The software requirements to use this tool are: Graphviz \cite{graphviz} 
	and the Erlang/OTP framework \cite{erlang}. 
	Both systems are free.
	We also provide a Docker file \cite{docker} for those environments where those requirements cannot be fulfilled.

	The rest of this section describes the functionality of \texttt{pn$\_$slicer}.
	
	\subsection{Functionality of \texttt{pn$\_$slicer}}
	\texttt{pn$\_$slicer} can be used in two ways: (i) it computes a slice with one specified slicing algorithm, or (ii) it uses all slicing algorithms to extract different slices, then it checks what properties of the original net are preserved by the slices and, finally, it outputs those slices that preserve a specified property set.
	
	
	
	\begin{lstlisting}[basicstyle=\ttfamily\scriptsize, float=*, frame=single, caption={\texttt{pn$\_$slicer} command format.}, label=lst:pnslicer-format]
	$ pn_slicer PNML_FILE SLICING_CRITERION [PROPERTY_LIST | ALGORITHM] [-json]
	\end{lstlisting}
	
	\begin{lstlisting}[basicstyle=\ttfamily\scriptsize, float=*, frame=single, caption={\texttt{pn$\_$slicer} command usage.}, label=lst:pnslicer-usage]
	$ pn_slicer pn_example.xml "P6,P9" "conflict_free"
	Petri net named pn_example successfully read.
	Slicing criterion: [P6, P9]
	1.- Llorens et al's slicer (maximal) -> Reduction: 9.09 %
	2.- Rakow's slicer CTL -> Reduction: 0.00 %
	3.- Yu et al's slicer -> Reduction: 13.64 %
	4.- Rakow's slicer safety -> Reduction: 4.55 %
	\end{lstlisting}

	
	\noindent Listing \ref{lst:pnslicer-format} shows the \texttt{pn$\_$slicer} command format, where {\small\texttt{SLICING\_CRITERION}} is a quoted list of places separated with commas.
	{\small\texttt{PROPERTY\_LIST}} is optional: APT properties can be used.
	The interested reader can consult the list of valid APT properties in the GitHub's repository of our tool.
	{\small\texttt{ALGORITHM}} is also optional. Whenever no algorithm is specified, all algorithms are used.
	
	For example, we can slice a net \texttt{pn\_example.xml} with respect to the slicing criterion \texttt{"P6,P9"} ensuring that the {\tt \small conflict\_free} property is preserved. This produces the slices shown in Figure~\ref{fig:slicingResults}. The command used is shown in Listing \ref{lst:pnslicer-usage}.
	
	
	
	The slices produced are saved in a file named {\small\texttt{output/} \texttt{<PNML\_NAME>\_<OUTPUT\_NUMBER>.pnml}}, where {\small\texttt{<OUTPUT\_NUMBER>}} indicates the algorithm used.
	For instance, Rakow's CTL$^{*}_{-x}$ slice generated with the command shown in Listing \ref{lst:pnslicer-usage} can be found at {\small\texttt{output/example\_2.pnml}} and at {\small\texttt{output/example\_2.pdf}}. Moreover, if we activate the flag {\small\texttt{-json}}, the tool also generates a JSON output with more information.
	
	The generated Petri nets can be exported in APT, standard PNML \cite{pnml} 
	(compatible with PIPE5 \cite{pipe5}), 
	DOT and over 50 additional formats provided by Graphviz.

	\subsection{Other Petri net slicing implementations}\label{sec_implementations}
	
	
	Different implementations of the main slicing algorithms can be found in the literature \cite{Rak12,Khan15,Yu15}, but none of them is publicly available.
	For instance, in \cite{Rak11,Rak12}, Rakow presented an empirical evaluation of her two slicing algorithms. However, the implementation of the algorithms is not public and it is not described, so the empirical evaluation is neither replicable nor comparable with other algorithms.
	In \cite{Khan14,Khan15}, Khan shows $\mathit{SLAP}_{n}$, a tool for slicing Petri nets and
	Algebraic Petri nets (APN). The $\mathit{SLAP}_{n}$ tool is an Eclipse plugin. This tool draws an unfolded APN model or a Petri net model and allows us to write properties in the form of temporal formulas. From these formulas, the tool automatically extracts criterion places; and when the user
	chooses a slicing algorithm, the sliced APN or Petri net model is generated. The authors only
	implemented static slicing algorithms in the first version of this tool: APN Slicing, Abstract Slicing, Safety Slicing, and Liveness Slicing. This tool is also not publicly available but we know that the authors are working on a more stable tool, which is under development and it is not published yet.
	
	Yu et al. describe in \cite{Yu15} a slicing tool consisting of three components: a graphical
	editor, a Petri-net executor, and a slicer. The tool was implemented in C$\#$ with GDI+.
	The slicer implements the dynamic slicing processes: (1) setting
	the initial marking and the slicing criterion, (2) generating the SDG, and (3) obtaining
	the slice. Finally, to analyse the states of a modelled system, the tool allows us to generate the reachability marking graphs of the original Petri net and also of the dynamic slice.
	Visual interfaces have been developed to show these operations and results. Unfortunately,
	it is not maintained anymore, and the URL of the tool is broken.
	

	\section{Empirical evaluation}\label{Sec-Evaluation}
	
	
	In order to empirically compare the implemented Petri net slicing algorithms, we performed a number of experiments.
	To conduct the evaluation we selected Petri nets from
	the benchmark suite \textit{Model Checking Contest @ Petri Nets 2017}.
	Firstly, we randomly selected 1--5 places to define a slicing criterion. This was done 20 times for each benchmark (thus, 860 slicing criteria were produced).
	Then, from each triple (slicing algorithm, Petri net, slicing criterion) we extracted the corresponding slice (we used the five slicing algorithms integrated into \texttt{pn$\_$slicer}).
	To ensure the validation and replicability of our experiments, we made available all data at:
	\url{https://github.com/tamarit/pn_suite}.
	In folder {\tt examples/mcc\_models}, all benchmarks are classified by year.
	In folder {\tt data}, all slicing criteria are classified by year together with the reports of the statistical analysis. 
	
	We executed all benchmarks using the same configuration in the same hardware
	(Apple M1, 8-core (4 performance, 4 efficiency), with 8GB RAM)
	to evaluate their performance. The experiments were conducted ensuring that the only active process was \texttt{pn$\_$slicer} (to avoid interferences, the other processes were killed or stopped).

	Table~\ref{tab:benchsProperties} summarises the results of the statistical analysis.
	This table compares the six slicing algorithms, one in each column:
    Llorens et al.'s algorithm 1 in \cite{Llo08} ({\tt L\cite{Llo08}}),
	Llorens et al.'s maximal slicing ({\tt LM}),
	Llorens et al.'s minimal slicing ({\tt Lm}),
	Rakow's CTL$^{*}_{-x}$ slicing ({\tt RC}),
	Rakow's Safety slicing ({\tt RS}), and
	Yu et al.'s slicing ({\tt Y}).
	In each row, the best value is in bold.

\begin{table}[h!]
		\caption{Empirical evaluation of the five slicing algorithms: performance and efficiency statistics.}
		\label{tab:benchsProperties}
		\centering
		{
			\scriptsize
			\begin{tabular}{|l|l|c|c|c|c|c|c|}
				\hline
				\multicolumn{2}{|l|}{\texttt{Performance and efficiency}} & \texttt{L\cite{Llo08}} & \texttt{LM} & \texttt{Lm} & \texttt{RC} & \texttt{RS} & \texttt{Y}\\
				\hline
				\hline
				\multicolumn{2}{|l|}{$num\_places$} & 91.64\% & 91.60\% & \textbf{16.80\%} & 92.97\% & 91.94\% & 67.64\% \\
				\multicolumn{2}{|l|}{$num\_tokens$} & 98.70\% & 98.20\% & \textbf{46.40\%} & 99.57\% & 98.90\% & 98.70\% \\
				\multicolumn{2}{|l|}{$num\_arcs$} & 92.21\% & 92.19\% & \textbf{6.34\%} & 93.66\% & 88.68\% & 47.71\% \\
				\multicolumn{2}{|l|}{$num\_transitions$} & 91.85\% & 91.82\% & \textbf{7.06\%} & 93.91\% & 88.34\% & 49.40\% \\
				\hline
				\multicolumn{2}{|l|}{Size (w.r.t. the original net)} & 91.83\% & 91.80\% & \textbf{11.39\%} & 93.15\% & 88.60\% & 56.65\% \\
				\multicolumn{2}{|l|}{Time (runtime in milliseconds)} & 11.39 & 13.00 & 749.87 & 5.69 & \textbf{5.41} & 60.04 \\
				\hline
			\end{tabular}
		}
	\end{table}
	
	We calculated the size of the slice compared with the original Petri net in terms of the number of places, tokens, arcs, and transitions (see the first four rows of Table \ref{tab:benchsProperties}). This can help us to know what specific dimensions have been reduced (transitions, places, etc.). 
	We also measured the size of the slices (considering transitions and places) with respect to the size of the respective original nets. This is shown in row {\tt Size}. As an average, the different algorithms reduced between 6.85\% and 88.61\% the size of the original Petri net. The reader should note that each algorithm has a different purpose. Therefore, as previously explained, the sizes of their slices are not directly related.
	
	Finally, the mean runtime used to produce a slice is shown in row {\tt Time}.
	There, we can see that the computation of slices is a relatively efficient process ($<1$ s. in all benchmarks). Comparatively, practically all algorithms showed similar runtimes ([5,13] ms.) except for Llorens et al.'s minimal slicing and Yu et al.'s algorithms, which showed a runtime of two and one order of magnitude more, respectively. These algorithms produce significantly smaller slices than the others because they explore more paths to the slicing criterion. This justifies that its runtime is higher. We also want to remark that Yu et al.'s slicing algorithm produced an average size smaller than the others (except for Llorens et al.'s minimal slicing) but Yu et al.'s algorithm does not guarantee that the slice produced can contribute tokens to all the places in the slicing criterion (while the other algorithms do guarantee this property). It is also important to highlight that the implementation of Llorens et al.'s minimal slicing algorithm does not check whether the slice exists (i.e., it does not implement lines 3-5 of Algorithm \ref{alg:Slicing-PN3}). This would imply to build the coverability tree \cite{KM69}, which is a costly operation. In the implementation, however, the algorithm uses a timeout to stop searching for the minimal slice. This significantly reduces the time complexity of the algorithm in most of the cases.

	\section{Conclusions}\label{Sec-Conclusions}
	
	We have presented two slicing algorithms that can be useful in the debugging and specialization of Petri nets. The first algorithm (called maximal contributing slicing algorithm) extracts from a Petri net all parts (places and transitions) that can contribute tokens to the slicing criterion. This algorithm
	is an improvement of the algorithm proposed by Llorens et al. in \cite{Llo08}. The performance of the forward slicing phase in the new algorithm is always monotonically more efficient than the previous one. Moreover, the new algorithm uses a final postprocess to remove useless parts of the initial Petri net. This algorithm is useful for Petri net comprehension and debugging because the slice always contains all the causes that produced an error in the slicing criterion.
	
	The second algorithm (called minimal contributing slicing algorithm) is a new approach that collects the places and transitions needed to fire the shortest transition sequence that contributes tokens to some place in the slicing
criterion. This algorithm is useful for component extraction and reuse. A variant of this algorithm has also been defined to extract the smallest slice of a Petri net.
	
	We have provided a notion of maximality and minimality of Petri net slices, and we have formally proven that the first algorithm is maximal and the second algorithm is minimal.
	
	These algorithms together with other algorithms (for which there was not a public implementation) have been implemented and integrated into a tool called
	\texttt{pn$\_$slicer}. The implementation of all the algorithms produced a synergy that is useful and necessary, since they have different purposes, and they retain different properties in their slices.
	
	
	Implementing these algorithms produced another interesting result: we have been able to fairly evaluate and compare them. The comparison is fair because all of them have been evaluated with the same benchmarks, slicing criteria, implementation language, and hardware configuration. Previous comparisons based on the results reported in the papers were totally imprecise and unfair because each paper used different benchmarks in its evaluation.
	
	Furthermore, before our implementation, some of the discussed algorithms had not been implemented (they were theoretical results and a publicly available implementation was missing). Therefore, the efficiency, scalability, and performance of these tools were unknown.
	Our empirical evaluation properly compared these tools reporting measures about average runtimes and sizes of the produced slices.
	
	Our work offers another side result: a web system that allows researchers to freely test our tool and extract slices with different algorithms without installing the system.
	All our results, including the tools and experiments, are open-source, free, and publicly available.
	
	%

	
	\subsection*{Acknowledgments}
	
	We thank the authors of the \emph{Model Checking Contest} benchmarks for their work. We thank the reviewers of the \emph{Fundamenta Informaticae} journal for their useful and constructive feedback.

\end{document}